\documentclass{ws-ijmpcs}

\usepackage{amssymb,amsmath,mathrsfs,tikz,graphicx}

\usepackage{lscape,chngpage}

\usepackage{rotating,wasysym}

\def\unwarrow{\nwarrow\joinrel\joinrel\uparrow}

\def\dsearrow{\downarrow\joinrel\joinrel\searrow}
\def\dotsearrow{\hspace{-4mm}\dot{\hspace{4mm}\searrow}}
\def\dotswarrow{\hspace{4mm}\dot{\hspace{-4mm}\swarrow}}

\def\Journal#1#2#3#4{{#1} {\bf #2}, #3 (#4)}

\def\CQG{\em Class. Quantum Grav.}
\def\PRD{\em Phys. Rev. D }
\def\GRG{\em Gen. Rel. Grav.}

\def\PRL{\em Phys. Rev. Lett.}
\def\JMP{\em J. Math. Phys.}

\def\espaitemps{({\cal V},g)}
\def\varietat{{\cal V}}

\def\xiv{\vec \xi }

\def\lie{{\pounds}}

\def\AH{\mbox{AH}}

\def\S{\Sigma}
\def\di{{\rm div}}
\def\sg{{\rm sign}}

\def\doo{d \Omega^2_{n-2}}

\def\scri{\mathscr{J}}
\def\B{\mathscr{B}}
\def\R{\mathscr{R}}

\def\be{\begin{equation}}
\def\ee{\end{equation}}
\def\bea{\begin{eqnarray}}
\def\eea{\end{eqnarray}}
\def\bean{\begin{eqnarray*}}
\def\eean{\end{eqnarray*}}

\begin{document}

\markboth{Jos\'e M M Senovilla}
{Trapped surfaces}

%
\catchline{}{}{}{}{}
%

\title{TRAPPED SURFACES}

\author{JOS\'E M. M. SENOVILLA}

\address{F\'{\i}sica Te\'orica, Universidad del Pa\'{\i}s Vasco, Apartado 644\\
48080 Bilbao, Spain\\
josemm.senovilla@ehu.es}



\maketitle

\begin{history}
\received{Day Month Year}
\revised{Day Month Year}
\end{history}

\begin{abstract}
I review the definition and types of (closed) trapped surfaces. Surprising global properties are shown, such as their ``clairvoyance'' and the possibility that they enter into flat portions of the spacetime. Several results on the interplay of trapped surfaces with vector fields and with spatial hypersurfaces are presented. Applications to the quasi-local definition of Black Holes are discussed, with particular emphasis set onto marginally trapped tubes, trapping horizons and the boundary of the region with closed trapped surfaces. Finally, the {\em core} of a trapped region is introduced, and its importance discussed.
\end{abstract}

\ccode{PACS numbers: 04.70.Bw}

\section{Introduction}
Black Holes (BH) are fundamental physical objects predicted classically in General Relativity (GR) which show a very deep relation between Gravitation, Thermodynamics, and Quantum Theory\cite{W}. Classically, the characteristic feature of a BH is its event horizon EH: the boundary of the region from where one can send signals to a far away asymptotic external region. This EH is usually identified as the surface of the BH and its area to the entropy. Unfortunately, the EH is essentially a global object, as it depends on the {\em whole} future evolution of the spacetime. Thus, EHs are determined by future causes, they are {\em teleological}, see\cite{AK1,AG,BS1,B,E,Haji} and references therein. 

However, it is important to recognize the presence of a BH locally, for instance in numerical GR\cite{T}, in the 3+1 or Cauchy-problem perspective of GR\cite{JVG} (see also Jaramillo's contribution in this volume), or in Astrophysics. In the former two, one needs to pinpoint when the BH region has been entered, while in the latter there are so many candidates to real BHs that a precise meaning of the sentence ``there is a BH in the region such and such'' is required. Of course, this meaning cannot rely on the existence of an EH as the real BHs undergo evolutionary processes and are usually dynamical. 

Over the recent years, there has been a number of efforts to give a general quasi-local description of a dynamical black hole.
In particular, the case has been made for quasi-local objects called Marginally Trapped Tubes (MTT) or Trapping Horizons (TH)\cite{AK1,AG,B,Hay,Hay1}, and their particular cases called {\em dynamical horizons} (DH)  \cite{AK1,AG} or Future Outer Trapping Horizons (FOTH)\cite{Hay,Hay1}. The underlying ideas were discussed in the 90s by Hayward\cite{Hay,Hay1}. MTTs are hypersurfaces foliated by closed (compact without boundary) marginally trapped surfaces. It is accepted that closed (marginally) trapped surfaces constitute the most important ingredient in the formation of BHs, so that the idea of using MTTs as the surface of BHs, and as viable replacements for the EH, looked very promising. This is one of the main reasons to study trapped surfaces, the subject of this contribution. 

Unfortunately, MTTs have an important problem: they are highly non-unique\cite{AG}. 
This manifests itself because the 2-dimensional Apparent Horizons\cite{HE,Wald} depend on the choice of a reference foliation of spacelike hypersurfaces.
Hence, another reasonable alternative, which is manifestly independent of any foliation, has been pursued and investigated more recently\cite{ABS,BS,BS1,S5}: the boundary
$\B$ of the future-trapped region $\mathscr{T}$ in spacetime (this is the region through which future-trapped surfaces pass). This is also a very natural candidate for the surface of a BH, as it defines a frontier beyond which MTTs and general trapped surfaces cannot be placed. I am going to show that $\B$ cannot be an MTT itself, and that it suffers from problems similar to that of EHs: it is unreasonably global.
This seems to be an intrinsic problem linked to a fundamental property of closed trapped surfaces: they are {\em clairvoyant}\cite{ABS,BS,BS1,S5}.

Recently, a novel idea\cite{BS1} has been put forward to address all these issues: the {\em core} of a trapped region, and its boundary. This is a minimal region which is indispensable to sustain the existence of closed trapped surfaces in the spacetime. It has some interesting features and it may help in solving, or better understanding, the difficulties associated to BHs. 

\section{Trapped surfaces: Definition and types}
In 1965, Penrose\cite{P2} introduced a new important concept: {\em closed trapped surfaces}. These are closed spacelike surfaces (usually topological spheres) such that their area tends to decrease locally along any possible {\em future} direction. (There is a dual definition to the past).
The traditional Black Hole solutions in GR, constituted by the Kerr-Newman family of metrics, have closed trapped surfaces in the region inside the Event Horizon.\cite{HE,Wald}
Actually, the existence of closed trapped surfaces is a fundamental ingredient in the singularity theorems\cite{P5,HE,HP,S}: if they form, then singularities will surely develop.

\subsection{Co-dimension two surfaces: Notation}
Let $\espaitemps$ be an $n$-dimensional
Lorentzian manifold with metric tensor $g$ of signature $(-,+,\dots ,+)$. 
A co-dimension two (dimension $n-2$) connected surface $S$ can be represented by means of its embedding $\Phi: S \longrightarrow \varietat$ into the spacetime $\varietat$ via some parametric equations $x^{\mu}=\Phi^{\mu}(\lambda^A)$ where $\{x^\mu\}$ are local coordinates in $\varietat$ ($\mu,\nu,\dots =0,1\dots ,n-1$), while $\{\lambda^A\}$ are local coordinates for $S$ ($A,B,\dots = 2,\dots ,n-1$).

The tangent vectors $\partial_{\lambda_A}$ are pushed forward to $\varietat$ to define the tangent vectors to $S$ as seen on $\varietat$ 
$$\vec{e}_A\equiv \Phi'(\partial_{\lambda^A}) \Longleftrightarrow e^{\mu}_A=\frac{\partial\Phi^{\mu}}{\partial\lambda^A}$$
and the first fundamental form of $S$ in $\espaitemps$ is defined as the pull-back of $g$:
$$
\gamma =\Phi^*g \Longrightarrow 
\gamma_{AB}(\lambda)=g|_S(\vec{e}_A,\vec{e}_B)=g_{\mu\nu}(\Phi)e^{\mu}_Ae^{\nu}_B
$$
From now on, $S$ will be assumed to be {\em spacelike} which means that $\gamma_{AB}$ is {\em positive definite}. Then, $\forall x\in S$ one can canonically decompose the tangent space $T_x\varietat$ as
$$
T_{x}\varietat =T_{x}S\oplus T_{x}S^{\perp}$$
which are called the  {\em tangent} and {\em normal} parts. 
In particular, we have\cite{O,Kr}
$$
\nabla_{\vec{e}_A}\vec{e}_B = \overline{\Gamma}^C_{AB}\vec{e}_C-\vec{K}_{AB}
$$
where $\nabla$ is the covariant derivative in $\espaitemps$, $\overline{\Gamma}^C_{AB}$ is the Levi-Civita connection associated to the first fundamental form in $S$ ($\bar\nabla_C\gamma_{AB}=0$) and 
$\vec{\bm{K}} : \, \mathfrak{X}(S) \times \mathfrak{X}(S) \longrightarrow \mathfrak{X}(S)^{\perp}$ 
 is called the {\em shape tensor} or {\em second fundamental form vector} of $S$ in $\varietat$. 
 [Here $\mathfrak{X}(S)$ ($\mathfrak{X}(S)^{\perp}$) is the set of vector fields tangent (orthogonal) to $S$.] Observe that $\vec{\bm{K}}$ is orthogonal to $S$.  $\vec{\bm{K}}$ measures the difference between the pull-back to $S$ of the covariant derivative of covariant tensor fields and the covariant derivative in $S$ of the pull-back of those tensor fields according to the formula
\be
\Phi^*(\nabla v) = \bar\nabla (\Phi^*v)+ v(\vec{\bm{K}}) \hspace{1cm} \Longleftrightarrow \hspace{1cm}
e^\mu_B e^\nu_A \nabla_\nu v_\mu =\overline\nabla_A \bar{v}_B +v_\mu K^\mu_{AB}
\label{proj}
\ee
where for all one-forms $v_\mu$ of $\varietat$ we write $\bar{v}_A\equiv v_\mu|_S\,  e^\mu_A=v_\mu(\Phi)\,  e^\mu_A$ for the components of its pull-back to $S$, $\bar v =\Phi^*(v)$.

A {\em second fundamental form} of $S$ in $\espaitemps$ relative to any normal vector field $\vec n\in \mathfrak{X}(S)^{\perp}$ is defined as
$$
K_{AB}(\vec n)\equiv n_{\mu} K^{\mu}_{AB} .
$$
For each $\vec n$, these are 2-covariant symmetric tensor fields on $S$.

\subsubsection{Mean curvature vector. Null expansions}
For a spacelike surface $S$ there are two {\em independent} normal vector fields, and one can choose them to be future-pointing and null everywhere, $ \vec k^\pm \in \mathfrak{X}(S)^{\perp}$ with
$$
k^+_{\mu}e^\mu_{A}=0, \, \, k^-_{\mu}e^\mu_{A}=0, \, \, k^+_{\mu}k^{+\mu}=0, \, \,
k^-_{\mu}k^{-\mu}=0. 
$$
By adding the normalization condition $k_{+\mu}k_{-}^{\mu}=-1$, there still remains the freedom
\be
\vec{k}^+ \longrightarrow \vec{k}'^+=\sigma^2 \vec{k}, \hspace{1cm}
\vec{k}^- \longrightarrow \vec{k}'^-=\sigma^{-2} \vec{k}^- \, .\label{norm}
\ee
One obviously has
$$
\vec{K}_{AB}=-K_{AB}(\vec k^-)\,\, \vec{k}^+ -K_{AB}(\vec{k}^+)\,\, \vec k^- \, .
$$

The {\em mean curvature vector} of $S$ in $\espaitemps$ is the trace of the shape tensor:
$$\vec H \equiv \gamma^{AB} \vec{K}_{AB}, \hspace{1cm} \vec H\in \mathfrak{X}(S)^{\perp}$$
and its decomposition in the null normal basis
$$\vec H \equiv \gamma^{AB} \vec{K}_{AB} = 
-\theta^- \vec{k}^+ -\theta^+\,\, \vec k^-$$
defines the so-called {\em future null expansions}: $\theta^{\pm}\equiv \gamma^{AB}K_{AB}(\vec k^\pm)$. Notice that, even though the null expansions are not invariant under the boost transformations (\ref{norm}), $\vec H$ is actually invariant.

\subsection{The trapped surface fauna: a useful symbolic notation}
The concept of trapped surfaces was originally formulated in terms of the signs or the vanishing of the null expansions\cite{P2}, and has remained as such for many years. This is obviously related to the causal orientation of the mean curvature vector. Thus, it has become clear over the recent years that the causal orientation of the mean curvature vector provides a better and powerful characterization of the trapped surfaces\cite{Kr,O,MS} \cite{S1}\cdash\cite{S4}.\footnote{The characterization by means of the mean curvature vector has permitted the extension of the classical singularity theorems to cases with trapped submanifolds of arbitrary co-dimension embedded in spacetimes of arbitrary dimension, see Ref.\refcite{GS}.}
Therefore, a symbolic notation for the causal orientation of $\vec H$ becomes very useful. Using an arrow to denote $\vec H$ and denoting the future as the upward direction and null vectors at 45$^o$ with respect to the vertical, the symbolic notation was introduced in\cite{S4}:
\begin{center}
\begin{tabular}{c|c}
$\vec H$ & Causal orientation \\
\hline
$\downarrow$  & past-pointing timelike\\
$\swarrow$ or $\searrow$ & past-pointing null ($\propto \vec k^+$ or $\vec k^-$) \\
$\leftarrow$ or $\rightarrow$  & spacelike\\
$\cdot$ & vanishes \\
$\nearrow$ or $\nwarrow$ & future-pointing null ($\propto \vec k^+$ or $\vec k^-$)\\
$\uparrow$ &  future-pointing timelike\\
\end{tabular}
\end{center}

A surface is said to be {\em weakly future-trapped} (f-trapped from now on) if the mean curvature vector is future-pointing all over $S$, similarly for weakly past-trapped. The traditional f-trapped surfaces have a timelike (non-vanishing) future-pointing $\vec H$ all over $S$, while the marginally f-trapped surfaces have $\vec H$ proportional to one of the null normal directions. Using the previously introduced notation ---if $\vec H$ changes causal orientation over $S$ the corresponding symbols are superposed\cite{S4}--- the main cases are summarized in the next table, where the signs of the null expansions are also shown.

\begin{table}[h]
\tbl{The main cases of future-trapped surfaces.}
{\begin{tabular}{c|c|l}
Symbol & Expansions & Type of surface \\
\hline
$\cdot$ & $\theta^+=\theta^-=0$ &stationary or minimal \\
\hline
$\uparrow$ & $\theta^+<0, \theta^-<0$ & f-trapped\\
\hline
\begin{sideways}{$\dotsearrow$}\end{sideways} & $\theta^+=0, \theta^-\leq 0$ & marginally f-trapped \\
\hline
\begin{sideways}\begin{sideways}{$\dotsearrow$}\end{sideways}\end{sideways} & $\theta^+\leq 0, \theta^-=0$ & marginally f-trapped \\
\hline

\begin{sideways}\begin{sideways}$\swarrow\dotsearrow$
\end{sideways}\end{sideways}\hspace{-7.5mm} \hspace{2mm}\raisebox{1mm}{$\uparrow$} & $\theta^+\leq 0, \theta^-\leq 0$ & weakly f-trapped \\
\hline
\end{tabular} \label{ta1}}
\end{table}
Sometimes, only the sign of one of the expansions is relevant. This may happen if there is a consistent or intrinsic way of selecting a particular null normal on $S$. Then, one can use the $\pm$-symbols to denote the preferred direction and define $\pm$-trapped surfaces. However, in the literature the preferred direction is usually declared to be ``outer'', and then the nomenclature speaks about ``outer trapped", no matter whether or not this outer direction coincides with any particular outer or external part to the surface. Thus, (marginally) $+$-trapped surfaces are usually referred to as (marginally) outer trapped surfaces ((M)OTS) and similarly for the $-$ case. The main possibilities are summarized in Table \ref{ta2}.
\begin{table}[h!]
\tbl{The main cases of $+$-trapped, also ``outer trapped'', surfaces.}
{\begin{tabular}{c|c|l}
Symbol & Expansions & Type of surface \\
\hline
\raisebox{-1.5mm}{$\leftarrow$}\hspace{-3mm}$\unwarrow$ & $\theta^+<0$ & half converging, or $+$-trapped (or OTS)\\
\hline
& & \\
\raisebox{-2mm}{$\dotswarrow$ \raisebox{4mm}{$\nearrow$}} & $\theta^+=0$ & null dual or M+TS (or MOTS) \\ 
& & \\
\hline
$\raisebox{2mm}{$\leftarrow$} \hspace{-4.1mm} \stackrel{\begin{sideways}\begin{sideways}$\swarrow\dotsearrow$
\end{sideways}\end{sideways}\hspace{-4mm} \raisebox{1mm}{$\uparrow$}}
{\swarrow} $ &
$\theta^+\leq 0$ & weakly $+$-trapped (W+TS or WOTS)\\
\hline
\end{tabular} \label{ta2}}
\end{table}
Important studies concerning these surfaces, and in particular MOTS, have been carried out recently in \cite{AMS,AMS1,AM,CM}, with relevant results for black holes and the existence of MTTs.

For completeness, I also present the characterization and symbols of some other distinguished types  of surfaces, such as (weakly) untrapped surfaces or the null untrapped surfaces ---recently also named ``generalized apparent horizons"  in Ref.\refcite{BK}--- in Table \ref{ta3}.
\begin{table}[h!]
\tbl{Miscellaneous surfaces.}
{\begin{tabular}{c|c|l}
Symbol & Expansions & Type of surface \\
\hline
$\rightarrow$ or $\leftarrow$ & $\theta^+\theta^- <0 $ & untrapped \\
\hline
\raisebox{-3mm}{\begin{sideways}$\swarrow\dotsearrow$
\end{sideways}\hspace{-3.8mm} \raisebox{2.6mm}{$\rightarrow$}} or  
$\leftarrow$\hspace{-2.8mm}\raisebox{-3mm}{\begin{sideways}\begin{sideways}\begin{sideways}
$\swarrow\dotsearrow$\end{sideways}\end{sideways}\end{sideways}}
& $\theta^+\geq 0, \theta^-\leq 0$ or $\theta^+\leq 0, \theta^-\geq 0$ & weakly untrapped \\
\hline
$\dsearrow$\hspace{-3mm}\raisebox{1.2mm}{$\rightarrow$} & $\theta^+> 0$ & half diverging or $+$-untrapped (or outer untrapped)\\
\hline
\begin{sideways}\begin{sideways}$\raisebox{2mm}{$\leftarrow$} \hspace{-4.1mm} \stackrel{\begin{sideways}\begin{sideways}$\swarrow\dotsearrow$
\end{sideways}\end{sideways}\hspace{-4mm} \raisebox{1mm}{$\uparrow$}}
{\swarrow} $\end{sideways}\end{sideways} & $\theta^+\geq 0$ & weakly $+$-untrapped (or weakly outer untrapped) \\
\hline
$\stackrel{\begin{sideways}\begin{sideways}$\swarrow\dotsearrow$
\end{sideways}\end{sideways}\hspace{-4mm} \raisebox{2mm}{$\uparrow$}}
{\hspace{3mm}\swarrow\!\downarrow\! \searrow} $ & $\theta^+\theta^- \geq 0$ &
$\varhexstar$-surfaces\cite{S4} \\
\hline
$\stackrel{\begin{sideways}\begin{sideways}$\swarrow\dotsearrow$
\end{sideways}\end{sideways}}
{\swarrow \searrow}$ &
$\theta^+\theta^- = 0$ & null $\varhexstar$-surfaces \\
\hline
\raisebox{-3mm}{\begin{sideways}$\swarrow\dotsearrow$\end{sideways}} & $\theta^+\theta^-=0$ and $\theta^+\geq 0, \theta^-\leq 0$ & null untrapped or generalized apparent horizon\cite{BK} \\
\hline
\end{tabular} \label{ta3}}
\end{table}
The last type of surface shown in Table \ref{ta3} was proposed as a viable replacement\cite{BK} for marginally trapped surfaces in a new version of the Penrose inequality\cite{P6,M}. However, this version cannot hold as a recent counterexample has been found in Ref.\refcite{CM1}. Nevertheless, it is known that there cannot be null untrapped closed surfaces embedded in spacelike hypersurfaces of Minkowski spacetime\cite{Kh}. An important question related with these issues and raised in Ref.\refcite{MS} is whether or not there can be any closed $\varhexstar$-surfaces in Minkowski spacetime ---or more generally, in stationary spacetimes.

\subsection{A useful formula for the scalar $H_\mu H^\mu$.}
Now, I am going to present a formula for the norm of the mean curvature vector associated to a given family of co-dimension-2 surfaces.\cite{S1} This allows one to ascertain which surfaces within the family can be trapped, marginally trapped, etcetera. Assume you are given a family of 
$(n-2)$-dimensional spacelike surfaces $S_{X^a}$ described by 
$$\{x^a=X^a\}, \hspace{1cm} a,b,\dots =0,1$$
where $X^a$ are constants and $\{x^{\alpha}\}$ local coordinates in $\espaitemps$. 
Locally, the line-element can be written as
\be
ds^2=g_{ab} dx^adx^b+2g_{aA}dx^adx^A+g_{AB}dx^Adx^B
\label{sdg}
\ee
where $g_{\mu\nu}(x^{\alpha})$ and $\det g_{AB} >0$.
There remains the freedom
\be
x^a \longrightarrow x'^a=f^a(x^b), \,\,
x^A \longrightarrow x'^A=f^A(x^B,x^c) \label{coord}
\ee
keeping the form (\ref{sdg}) and the chosen family of surfaces. 

The imbedding $\Phi$ for the surfaces $S_{X^a}$ is
$$
x^a=\Phi^a(\lambda)= X^a=\mbox{const.}, \hspace{.5cm} x^A=\Phi^A(\lambda)=\lambda^A$$
from where the first fundamental form for each $S_{X^a}$ reads
$$
\gamma _{AB}=g_{AB}(X^a,\lambda^C)
$$ 
while the future null normal one-forms become
$$
\bm{k}^{\pm}=k^{\pm}_bdx^b\vert_{S_{X^a}}, \hspace{2mm}
g^{ab}k^{\pm}_ak^{\pm}_b=0, \hspace{2mm}
g^{ab}k^{+}_ak^{-}_b=-1
$$
Notice that $g^{ab}$ are components of $g^{\mu\nu}$ and therefore $(g^{ab})$ is not necessarily the inverse of $(g_{ab})$ !

Set
$$
G\equiv +\sqrt{\det g_{AB}} \equiv e^U, \hspace{3mm} \bm{g}_{a}\equiv 
g_{aA}dx^A 
$$
where $\bm{g}_a$ are considered to be two one-forms, one for each $a=0,1$. When pull-backed to $S_{X^a}$ they read $\bar{\bm{g}}_{a}=g_{aA}(X^{a},\lambda^C)d\lambda^{A}$.
A direct computation\cite{S1} provides the null expansions ($f_{,\mu}= \partial_{\mu} f$)
$$
\theta^{\pm}=\left.k^{\pm 
a}\left(\frac{G_{,a}}{G}-\frac{1}{G}(G\gamma^{AB}g_{aA})_{,B}\right)
\right\vert_{S_{X^a}} \, .
$$
and thereby the mean curvature one-form
\be
\fbox{$\displaystyle{
H_{\mu}=\delta^a_{\mu}\left(U_{,a}-\di \vec{g}_{a} \right)}$}
\ee
where $\di$ is the divergence operating on vector fields at each surface $S_{X^a}$.
These surfaces are thus trapped if and only if the scalar 
$$
\fbox{$\displaystyle{\kappa=-\left.g^{bc}H_{b}H_{c}
\right\vert_{S_{X^a}}}$}
$$
is positive, and a necessary condition for them to be marginally trapped is that $\kappa$ vanishes on the surface. Observe that one only needs to compute
the norm of $H_{a}$
{\it as if it were a one-form in the ``2-dimensional'' contravariant metric $g^{ab}$}. 

\section{Horizons: MTTs, FOTHs and Dynamical Horizons}
\subsection{$\{S_{X^a}\}$--horizons}
\label{subsec:horizons}
In the construction of the previous section, in general the mean curvature vector $\vec H$ of the $S_{X^{a}}$-surfaces will change its causal character at different regions of the spacetime. The hypersurface(s) of separation
$${\cal H}\equiv \{\kappa =0\}, \hspace{1cm} ``S_{X^a}-horizon"$$
is a fundamental place in $\espaitemps$ associated to the given family of surfaces $S_{X^a}$ called the $S_{X^a}-$horizon. It
contains the regions with marginally trapped, marginally outer trapped, null untrapped, and null $\varhexstar$-surfaces $S_{X^a}$
(as well as those parts of each $S_{X^a}$ where one of the $\theta^\pm$ vanishes). 

As an example and to show the applicability of the previous formulas, consider the 4-dimensional Kerr spacetime ($G=c=1$, $n=4$) in advanced (or Kerr) coordinates\cite{HE}
\bea
ds^2 = -\left(1-\frac{2Mr}{r^2+a^2\cos^2\theta}\right)dv^2+2dvdr
-2a\sin^2\theta d\varphi dr -\frac{4Mar\sin^2\theta}{r^2+a^2\cos^2\theta}d\varphi dv\nonumber\\
+\left(r^2+a^2\cos^2\theta\right)d\theta^2+\left(r^2+a^2+\frac{2Mar\sin^2\theta}{r^2+a^2\cos^2\theta}\right)\sin^2\theta d\varphi^2 \label{kerr}
\eea
where $M$ and $a$ are constants.
A case of physical interest is given by the topological spheres defined by constant values of $v$ and $r$, so that with the notation of the previous section one has $\{x^a\}=\{v,r\}$,  $\{x^A\}=\{\theta,\varphi\}$. The two one-forms $\bm{g}_a$ are
$$\bm{g}_{r}=-a\sin^2\theta d\varphi, \hspace{1cm} \bm{g}_{v}=-\frac{2Mar\sin^2\theta}{r^2+a^2\cos^2\theta} d\varphi$$
so that  $\di \vec{g}_{a}=0$.
On the other hand
$$
e^{2U}=\sin^2\theta [(r^2+a^2)(r^2+a^2\cos^2\theta) +2Mra^2\sin^2\theta],
$$
so that the mean curvature one-form becomes
$$
H_{a}dx^a=U_{,a}dx^a= \frac{r(r^2+a^2+r^2+a^2\cos^2\theta)+Ma^2\sin^2\theta}{(r^2+a^2)(r^2+a^2\cos^2\theta) +2Mra^2\sin^2\theta}dr
$$
from where one easily obtains
$$
\kappa = -\frac{(r^2-2Mr+a^2)}{(r^2+a^2\cos^2\theta)} \frac{\left[r(r^2+a^2+r^2+a^2\cos^2\theta)+Ma^2\sin^2\theta\right]^2}{\left[(r^2+a^2)(r^2+a^2\cos^2\theta) +2Mra^2\sin^2\theta\right]^2}
$$
ergo (for $r>0$) $\sg\, \kappa=-\sg \left(r^2-2Mr+a^2\right)$. Thus ${\cal H}$ are the classical Cauchy and event horizons of the Kerr spacetime\cite{HE}.

As another important example, to be used repeatedly later on in this contribution, consider general spherically symmetric spacetimes, whose line-element can always be cast in the form
$$
ds^2=g_{ab}(x^c)dx^adx^b+r^2(x^c)\doo
$$
where $\doo$ is the round metric on the $(n-2)$-sphere, $r$ is a function depending only on the $x^b$-coordinates called the area coordinate, and $\det g_{ab}<0$ so that $g_{ab}$ is a 2-dimensional Lorentzian metric. By taking the family of round spheres as the selected family of co-dimension two surfaces, so that they are given by constant values of $x^b$, the two one-forms $\bm{g}_a$ vanish so that
$$H_{a}= U_{,a}\propto \frac{r_{,a}}{r}$$

One can thus define the standard ``mass function''
$$
2 m(x^a)\equiv r^{n-3}\left(1-g^{bc}r_{,b}r_{,c}\right) 
$$
so that
$$
\kappa = -g^{ab}H_{a}H_{b} \propto -\left(1-\frac{2m}{r^{n-3}}\right)
$$
Hence, the round $(n-2)$-spheres are (marginally) trapped if $2m/r^{n-3}$ is (equal to) less than 1. The corresponding horizon
${\cal H}: r^{n-3} = 2m$ thus becomes the classical spherically symmetric apparent $(n-1)$-horizon AH \footnote{Observe that AH is considered here to be a hypersurface in spacetime, while the traditional ``apparent horizons'' are co-dimension two marginally trapped surfaces\cite{HE,Wald}. However, AH is foliated by such marginally trapped round spheres, so that it is a collection of apparent horizons}.

\subsection{MTTs, FOTHs, and dynamical horizons}
The apparent $(n-1)$-horizon AH in spherical symmetry, as well as some $S_{X^a}$-horizons, are examples of marginally trapped tubes (MTTs).
\begin{definition}
A marginally trapped tube is a hypersurface foliated by closed marginally f-trapped surfaces.
\end{definition}
This corresponds to the concept of ``future trapping horizons'' as defined by Hayward\cite{Hay}, who also introduced the concept of {\em future outer trapping horizon} by requiring that the vanishing null expansion $\theta^+=0$ becomes negative when perturbed along the other future null normal direction $\vec k^-$. In that case, the horizon is necessarily a spacelike hypersurface ---unless the whole second fundamental form $K_{AB}(\vec k^+)=0$ vanishes and also $G_{\mu\nu}k^{+\mu}k^{+\nu}=0$, where $G_{\mu\nu}$ is the Einstein tensor of the spacetime, in which case the horizon is null and actually a non-expanding or isolated horizon\cite{AK1}.

However, in many situations the ``outer'' condition in Hayward's definition is not required to obtain results about MTTs, and at the same time there are speculations that in dynamical situations describing the collapse to a realistic black hole the MTTs will eventually become spacelike. This led to the definition of {\em dynamical horizons} (DH), which are simply {\em spacelike} MTTs.

In spherical symmetry, the apparent $(n-1)$-horizons AH are the unique {\em spherically symmetric} MTTs. This does not only mean that they are the unique MTTs foliated by marginally trapped {\em round spheres}, but also that they are the only spherically symmetric hypersurfaces foliated by any kind of marginally trapped surfaces\cite{BS1}. As seen before, these preferred MTTs can be invariantly defined by
$$
\nabla_\mu r \nabla^\mu r=0 \Longleftrightarrow r^{n-3} = 2m \, .
$$
Despite this fact, it is known that MTTs, even in spherical symmetry, are not unique\cite{AG}, an explicit proof of this fact in generic spherically symmetric spacetimes can be found in Ref.\refcite{BS}, see Corollary \ref{coro} below. 

\section{The future-trapped region $\mathscr{T}$ and its boundary  $\B$}
The non-uniqueness of MTTs poses a difficult problem for the physics of black holes, and casts some doubts on whether MTTs provide a good quasi-local definition of the surface of a black hole. It is also a manifestation that the $(n-2)$-dimensional apparent horizons will depend on a chosen foliation in the spacetime. Thus, a reasonable alternative is to consider the boundary of the region containing f-trapped surfaces. Thus, following Refs.\cite{Hay,BS1} we define the following.
\begin{definition}
The future-trapped region $\mathscr{T}$ is defined as the set of points $x\in \varietat$ such that $x$ lies on a closed future-trapped surface.
We denote by $\B$ the boundary of the future trapped region $\mathscr{T}$:
$$
\B \equiv \partial \mathscr{T} .
$$
\end{definition}
$\mathscr{T}$ and $\B$ are invariant by the isometry group of the spacetime. More precisely\cite{BS1}
\begin{proposition}
If $G$ is the group of isometries of the spacetime $\espaitemps$, then $\mathscr{T}$ is invariant under the action of $G$, and the transitivity surfaces of $G$, relative to points of $\B$, remain in $\B$.
\end{proposition}
Therefore, in arbitrary spherically symmetric spacetimes, $\mathscr{T}$ and $\B$ have spherical symmetry. Actually the result is stronger.
\begin{proposition}\label{res:hypersurface}
In arbitrary spherically symmetric spacetimes, $\B$ (if not empty) is a spherically symmetric hypersurface without boundary.
\end{proposition}
\begin{figure}[h!]
\begin{center}
\includegraphics[height=7.5cm]{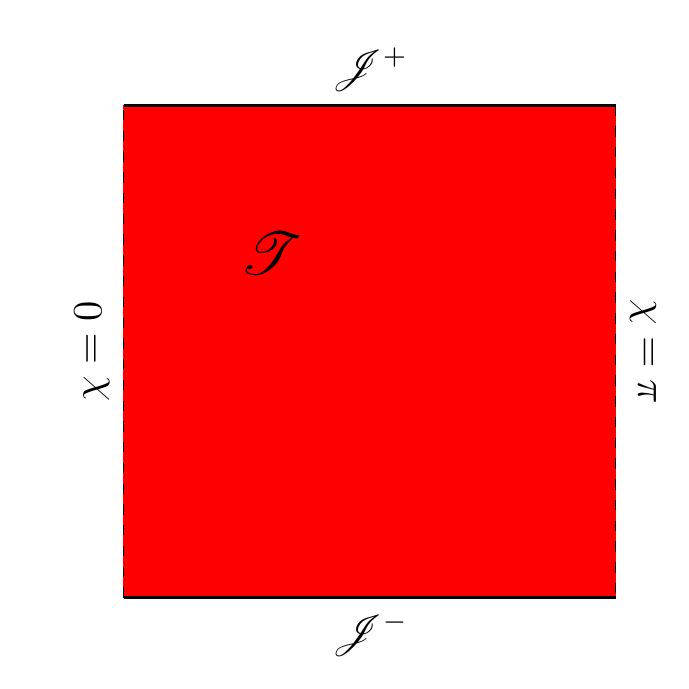}\end{center}
\caption{A Penrose conformal diagram for the de Sitter spacetime. The whole spacetime is coloured in red because there are f-trapped spheres passing through every point, so that $\mathscr{T}$ is the entire de Sitter spacetime. Thus, there is no boundary 
$\B$ for the f-trapped region. \label{f1} }
\end{figure}

As simple examples of these concepts, consider de Sitter spacetime.\cite{HE,Exact} It is well known that there are f-trapped round spheres in such spacetime, but given that it is a homogeneous spacetime then there must be such a f-trapped sphere through every point. Therefore, $\mathscr{T}$ is the whole spacetime and $\B =\emptyset$. This is represented in figure \ref{f1}.

An example with non-empty $\B$ is provided by the closed Robertson-Walker dust model with $\Lambda =0$\cite{HE,Exact}. In this case, $\mathscr{T}$ covers only `half' of the spacetime and the boundary $\B$ can be seen to correspond to the recollapsing time, which is a maximal hypersurface, as shown in fig.\ref{f2}, see Ref.\refcite{BS1}.
\begin{figure}[!ht]
\includegraphics[height=7.5cm]{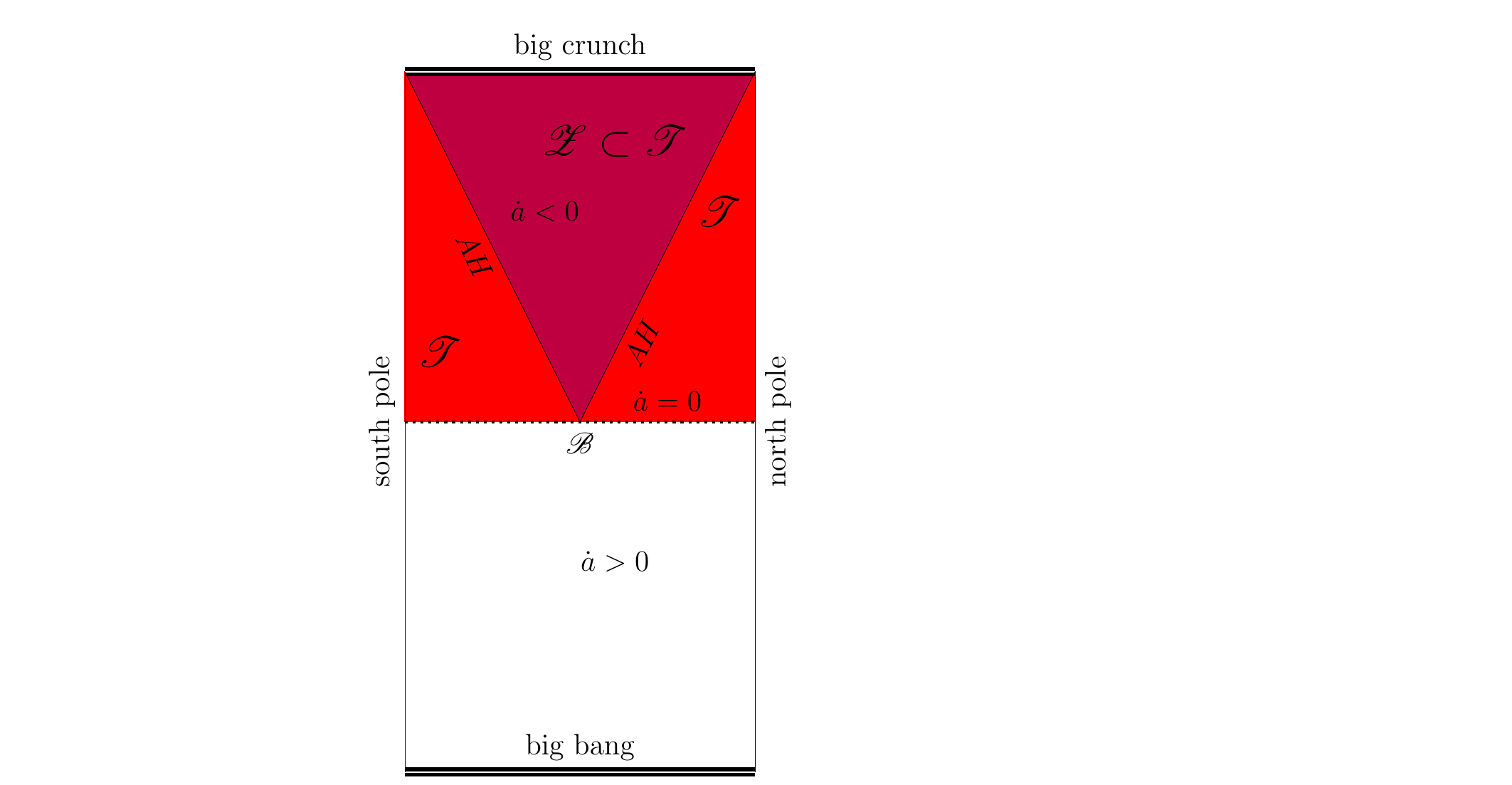}
\caption{Conformal diagram for the Robertson-Walker dust model. Closed f-trapped round spheres exist in the contracting phase, the coloured region, defined by $\dot a <0$, where $a$ is the scale factor. Thus, the boundary $\B$ is the hypersurface corresponding to the instant of recollapse, shown as a dotted horizontal line. An MTT corresponding to AH as defined in the text is also shown. This will be relevant for the definition of a core in Section \ref{sec:cores}.}
\label{f2}
\end{figure}

As a final example with $\B=\emptyset$, consider the flat spacetime, with line-element
$$ds^2 =-dt^2+dx^2+dy^2+dz^2$$
One can easily check\cite{S} that the family of surfaces $S_{x_0,t_0}\, : \{x=x_0,\hspace{2mm} e^{t_0-t}=\cosh z\}$ are f-trapped.  However, they are non-compact, see fig.\ref{f3}.
\begin{figure}[ht]
\begin{center}
\includegraphics[height=4cm]{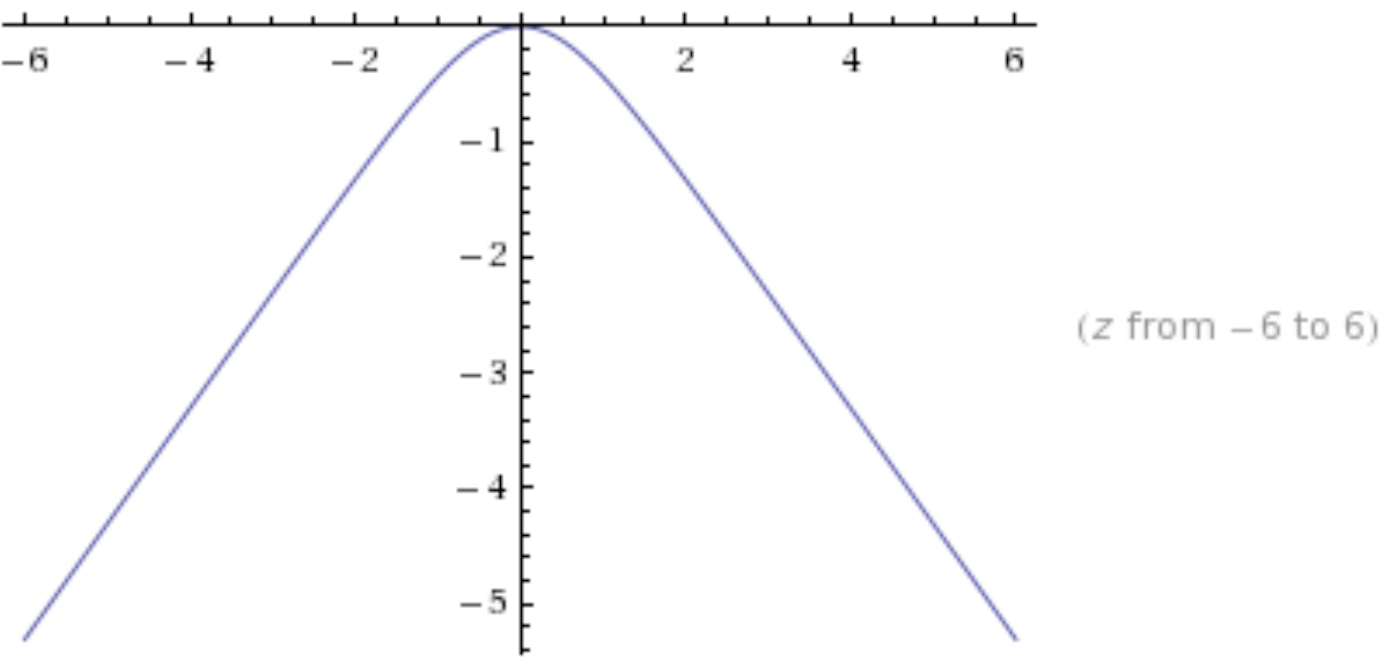}
\end{center}
\caption{A non-compact f-trapped surface in Minkowski spacetime. In spacetimes with selected foliations by hypersurfaces with some extrinsic properties ---such as the static time $t$ in flat spacetime---, all f-trapped surfaces have to ``bend down'' in time $t$, as shown ($t$ is the vertical coordinate). This property will prove to be very useful in some studies concerning black holes.}
\label{f3}
\end{figure}
It is a well-known result that there cannot be any {\em closed} trapped surface in flat spacetime. Actually, they are absent in arbitrary (globally) stationary spacetimes\cite{MS}. Therefore, in this case $\mathscr{T}=\emptyset$ and there is no boundary. There are no MTTs (such as AH) either. A conformal diagram is presented in fig.\ref{f4}.
\begin{figure}[h!]
\begin{center}
\includegraphics[height=5.5cm]{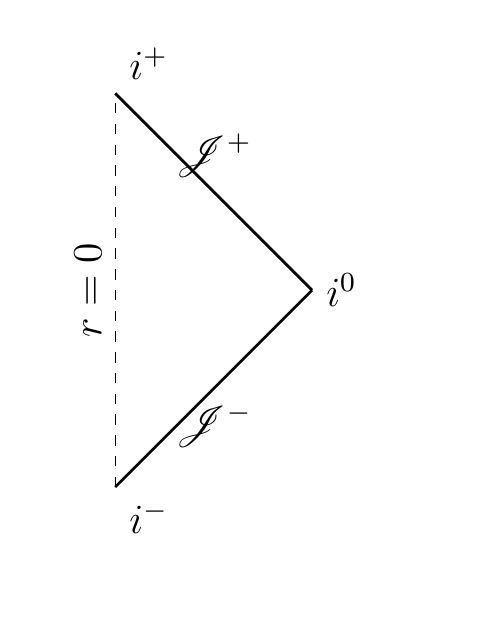}\end{center}
\caption{Flat (Minkowski) spacetime conformal diagram. There is no red region now ($\mathscr{T}=\emptyset$) because there are no closed (marginally) trapped surfaces in flat spacetime. Observe that the structure of infinity is given by two null hypersurfaces, $\scri^+$ and $\scri^-$, and three points $i^+$, $i^-$ and $i^0$. Spacetimes with a similar structure of infinity are called {\em asymptotically flat}. In flat spacetime all causal endless curves have an initial point at $\scri^-\cup i^-$, and a final point at $\scri^+\cup i^+$. This will not happen in general.}
\label{f4}
\end{figure}

\section{The event horizon, and its relation to $\B$ and AH.}
Consider now the Schwarzschild solution ($n=4$) in ``Eddington-Finkelstein'' advanced coordinates\cite{HE} (units with $G=c=1$), given by
\be
ds^2=-\left(1-\frac{2M}{r}\right)dv^2+2dvdr+r^2d\Omega^2 \label{sch}
\ee
This is (locally) the only spherically symmetric solution of the vacuum Einstein field equations. The mass function is now a constant $m=M$ representing the total mass, and $v$ is advanced (null) time. This metric is also the case $a=0$ (no rotation) of the Kerr metric (\ref{kerr}). From the above we know that the round spheres ---defined by constant values of $v$ and $r$--- are trapped if and only if $r<2M$. If $r=2M$ they are marginally trapped. Thus, the unique spherically symmetric MTT is given by AH, defined as
\begin{center}
AH :\hspace{1cm} $r-2M=0$  
\end{center}
One can actually prove that all possible closed f-trapped surfaces, be they round spheres or not, must lie completely inside the region $r<2M$. Therefore we also have $\mathscr{T}=\{r<2m\}$ and therefore, $\B=$AH. Even more, one can see, as shown in fig.\ref{f5}, that the spacetime is asymptotically flat, but that there are many future-endless causal curves that never reach future infinity, $\scri^+\cup i^+$. Therefore, there is a hypersurface separating those points which can send signals to infinity, from those which cannot. 
\begin{figure}[h!]
\begin{center}
\includegraphics[height=5.5cm]{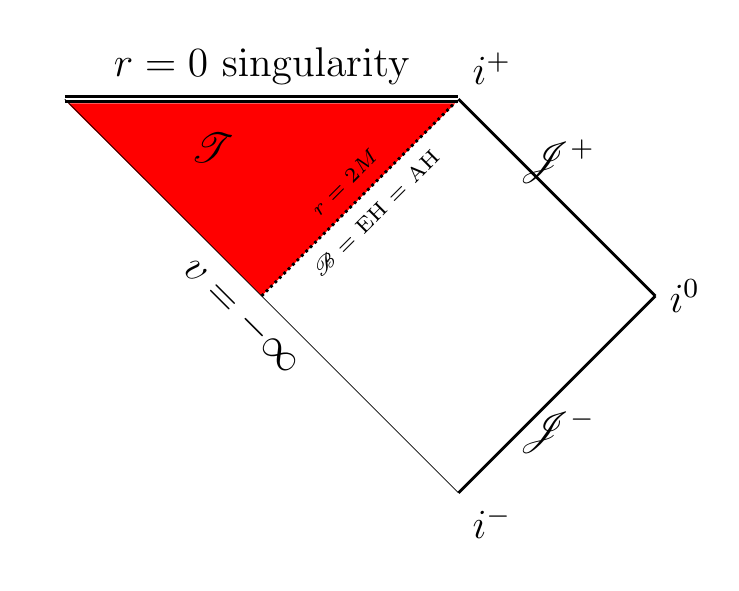}\end{center}
\caption{Conformal diagram for the Schwarzschild spacetime (\ref{sch}).  The red region is now confined by a boundary $\B$ which coincides with the unique spherically symmetric MTT, and with the event horizon EH.}
\label{f5}
\end{figure}
This separating membrane is called the {\em event horizon} (EH) in general, and in this case it happens to coincide with the spherically symmetric MTT and with the boundary of the f-trapped region
\begin{center}
AH = EH = $\B \,$ :\hspace{1cm} $r-2M=0$  .
\end{center}
However, this is an exceptional case, and these coincidences do not hold in general, dynamical, situations.

For general dynamical but {\em asymptotically flat} spacetimes (as remarked in fig.\ref{f4}, this is when the asymptotic region is ``Minkowskian'' with $\scri^{\pm}$ and $i^0$) one can define the region from where $\scri^+$ cannot be reached by any causal means. The boundary of this region is called the {\em Event Horizon} EH. (The formal definition is that EH is the boundary of the past of $\scri^+$). By definition, this is always a null hypersurface. However, and contrary to what happens in the Schwarzschild metric, it is not an MTT in general. 

To prove this claim, a simple example will suffice. Consider the imploding Vaidya spacetime\cite{V}, given in advanced coordinates by the line-element ($n=4$)
\be
ds^2=-\left(1-\frac{2m(v)}{r}\right)dv^2+2dvdr+r^2d\Omega^2 \label{vaidya}
\ee
Observe that this has exactly the same form as (\ref{sch}) but now the mass function depends on advanced null time $v$. The Einstein tensor of this metric is
$$
G_{\mu\nu}=\frac{2}{r^2}\frac{dm}{dv}\ell_\mu\ell_\nu, \hspace{1cm} \ell_\mu dx^\mu=-dv,\hspace{7mm}  \vec\ell =-\partial_{r} \hspace{3mm} (\ell^\mu\ell_\mu=0)
$$
which vanishes only if $dm/dv=0$. The null convergence condition (or the weak energy condition)\cite{HE,Wald} requires that $dm/dv\geq 0$. For simplicity, and to illustrate some important points concerning EHs and BHs, consider the following restrictions on the mass function
$$
m(v)=0 \hspace{3mm} \forall v<0; \hspace{1cm} m(v)\leq M <\infty \hspace{3mm} \forall v>0 .
$$
Then, the spacetime is flat for all $v<0$, while $M$ can be considered as the total mass.
The unique spherically symmetric MTT is now given by
\begin{center}
AH:\hspace{1cm} $r-2m(v)=0$
\end{center}
and it is easily checked that this hypersurface is spacelike whenever $dm/dv>0$ (and null at the regions with $dm/dv=0$). Therefore, EH is different from AH everywhere except possibly at a final asymptotic region with $m=M$=const. This proves that EH is not an MTT in general, see fig.\ref{f6}.
\begin{figure}[h!]
\begin{center}
\includegraphics[height=9.5cm]{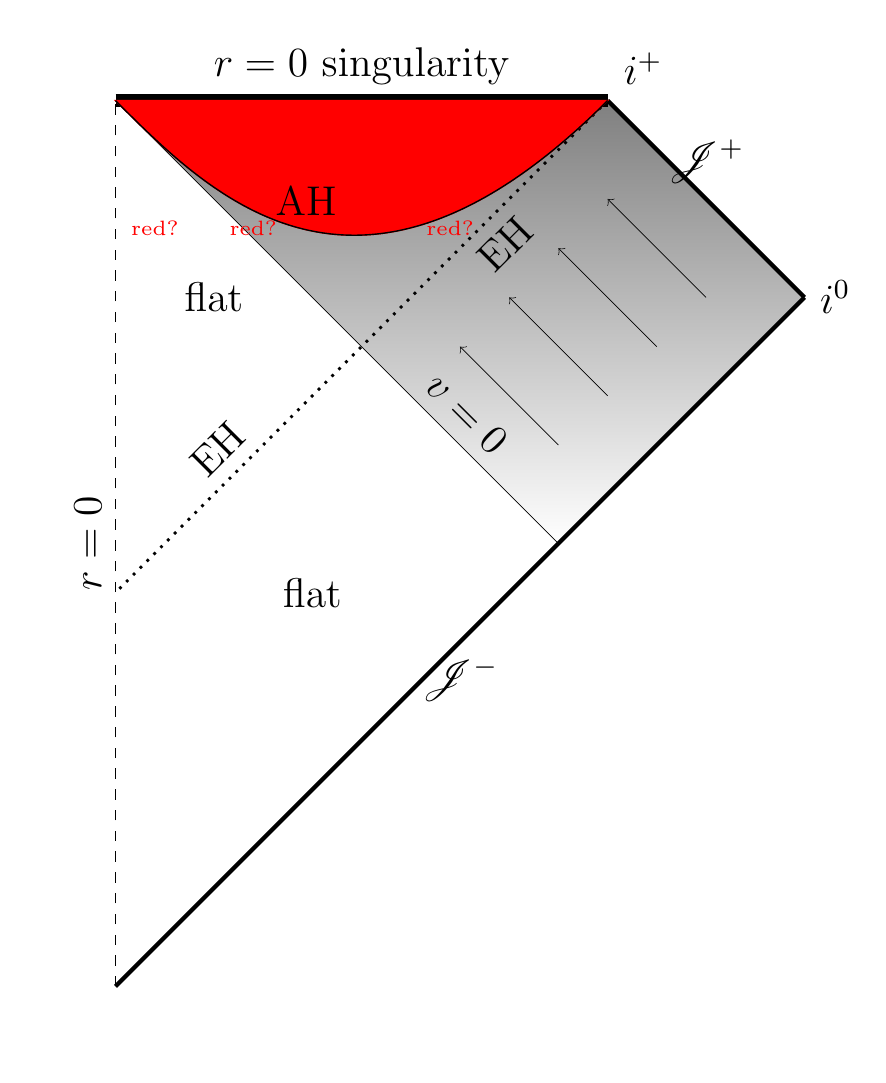}\end{center}
\caption{An imploding Vaidya spacetime. The spacetime is flat until null radiation flows in in spherical manner and produces a singularity when meeting at the center. This is ``clothed'' by an event horizon, as shown. However, this EH has a portion inside the flat zone of the spacetime, showing that it is ``aware'' of things that will happen in the future. The unique spherically symmetric MTT, denoted by AH, is a spacelike hypersurface (ergo a dynamical horizon) which never coincides with EH, and it approaches it asymptotically. One knows that there are closed f-trapped round spheres in the red region ($r<2m(v)$), so the question arises of whether or not AH is the boundary $\B$ of the f-trapped region. One can also wonder if the boundary $\B$ will actually be EH.}
\label{f6}
\end{figure}

In the figure one can also graphically see some of the global properties of event horizons, such as for example that they may start developing in regions whose whole past is flat. This is the teleological behaviour alluded to in the Introduction.

From the general results above, one knows that all round spheres in the region with $r<2m(v)$ are f-trapped. But now AH and EH do not coincide. Can there be any other f-trapped surfaces which extend outside AH? Will they be able to extend all the way down to EH? Put another way, one wants to know if the boundary $\B$ is EH, or AH, or neither. 
Ben-Dov proved\cite{BD} that the event horizon cannot be the boundary of closed f-trapped surfaces for the particular case of a shell of null radiation. This proves that EH $\neq \B$ in general.\footnote{Interestingly enough, he also proved that EH is the boundary for (marginally) {\em outer} f-trapped closed surfaces (MOTS) in the Vaidya spacetimes. This is a general conjecture due to Eardley \cite{E}.}
Numerical investigations\cite{SK}, however, were incapable of finding closed f-trapped surfaces to the past of the apparent 3-horizon AH. Thus, a natural question arises: is AH $=\B$ ? 

The answer is negative for general spherically symmetric spacetimes with inflowing matter and radiation, as discussed in the next section.

\section{Imploding spherically symmetric spacetimes in advanced coordinates: $\B\neq \AH$.}
The general 4-dimensional spherically symmetric line-element can be locally given in advanced coordinates by
$$
ds^2=-e^{2\beta}\left(1-\frac{2m(v,r)}{r}\right)dv^2+2e^\beta dvdr+r^2d\Omega^2 
$$
where $m(v,r)$ is the mass function. The future-pointing radial null geodesic vector fields ($k_{\mu}\ell^\mu =-1$) read
$$
\vec\ell = -e^{-\beta}\partial_r , \hspace{1cm} 
\vec k=\partial_v +\frac{1}{2}\left(1-\frac{2m}{r}\right)e^{\beta}\partial_r
$$
and the mean curvature vector for each round sphere (defined by constant values of $r$ and $v$) is:
\be
\vec H_{sph}=\frac{2}{r}\left(e^{-\beta}\partial_{v}+\left(1-\frac{2m}{r}\right)\partial_{r}\right)
\label{Hspheres}
\ee
so that setting $\vec k_{sph}^+ =\vec k$ and $\vec k_{sph}^- =\vec \ell$, the future null expansions for these round spheres become 
$$
\theta_{sph}^+ =\frac{e^{\beta}}{r}\left(1-\frac{2m}{r}\right), \hspace{1cm}
\theta_{sph}^-=-\frac{2e^{-\beta}}{r} 
$$
and the unique spherically symmetric MTT is given by $\AH : \, r-2m(r,v)=0\, (\Longleftrightarrow \theta_{sph}^+=0$), 
in agreement with previous calculations in Section \ref{subsec:horizons}.

AH can be timelike, null or spacelike depending on the sign of
$$
\left. \frac{\partial m}{\partial v}\left( 1-2\frac{\partial m}{\partial r}\right)\right|_{AH}
$$
In particular, AH is null (in fact it is an {\em isolated horizon}\cite{AK1}) on any open region where $m=m(r)$. This isolated horizon portion of AH, denoted by $\AH^{iso}$, is characterized by:
$$
\AH^{iso}
\equiv \AH \cap \{G_{\mu\nu} k^\mu k^\nu =0\}$$
An example (based on a Vaidya spacetime) is shown in the next figure \ref{f7}.
\begin{figure}[h!]
\begin{center}
\includegraphics[height=9.5cm]{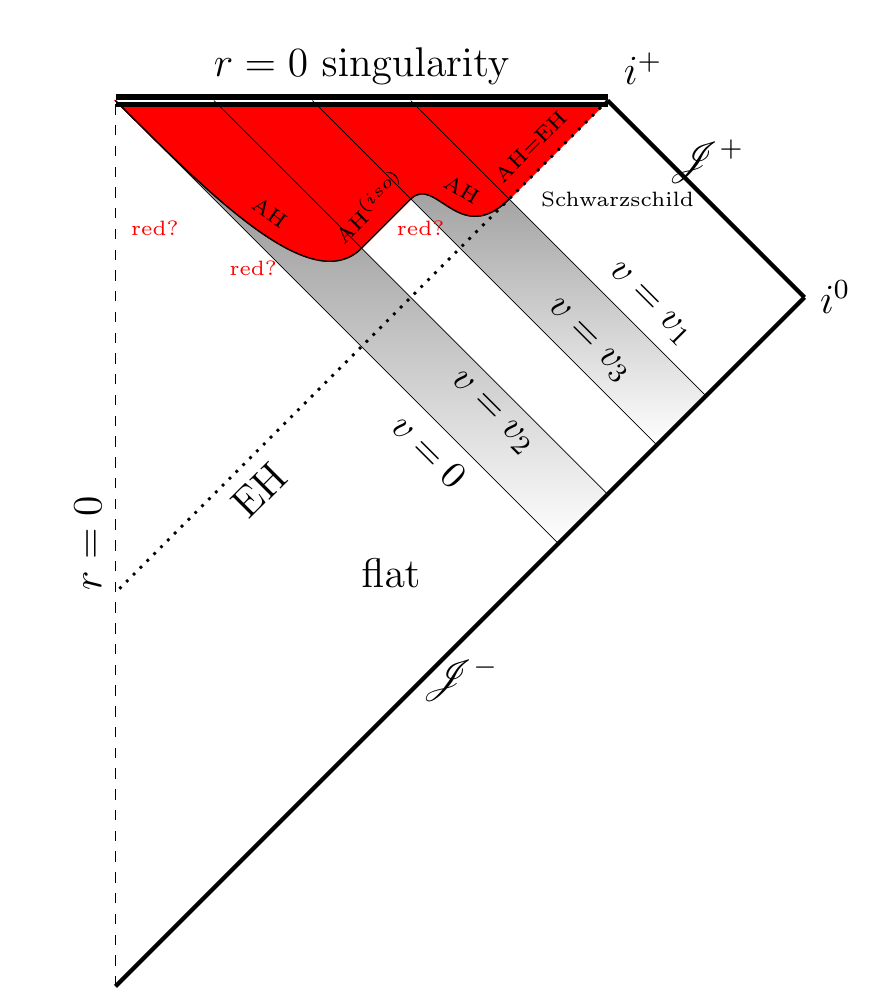}
\end{center}
\caption{A dynamical situation with non-empty $\AH^{iso}\backslash$EH. Null radiation flows for $v\in (0,v_{2})$, and then stops. There is no inflow of energy for $v\in (v_{2},v_{3})$ and then it comes in again from past null infinity for $v\in (v_{3},v_{1})$. From $v\geq v_{1}$ there is no more infalling energy and the spacetime settles down to a Schwarzschild black hole. The unique spherically symmetric MTT, denoted by AH, is spacelike for $v\in (0,v_{2})\cup (v_{3},v_{1})$, but it also possesses two portions of isolated-horizon type: one given by the part of EH with $v\geq v_{1}$, and the other one for $v\in (v_{2},v_{3})$, represented here by $\AH^{iso}$. The red portion is part of the f-trapped region $\mathscr{T}$, but one can prove (see main text) that there are closed f-trapped surfaces extending below $\AH\backslash\AH^{iso}$.}
\label{f7}
\end{figure}

\subsection{Perturbations on the spherical MTT}
In order to prove that there are closed f-trapped surfaces extending beyond AH one can use a perturbation argument, see Ref.\refcite{BS1}.
\begin{theorem}
In arbitrary spherically symmetric spacetimes, there are closed f-trapped surfaces (topological spheres) 
penetrating both sides of the apparent 3-horizon AH at any region where $G_{\mu\nu}k^\mu k^{\nu}|_{AH} \neq 0$.
\label{th:zigzag}
\end{theorem}
\begin{proof}
Recall that $\theta_{sph}^-=-e^{-\beta}\frac{2}{r} <0 , \hspace{2mm} \theta_{sph}^+=0$ on AH.
Perturb any marginally f-trapped 2-sphere $\varsigma\in$ AH along a direction $f\vec n$  orthogonal to $\varsigma$, where $f$ is a function on $\varsigma$ and 
$$
\vec n =-\left. \vec\ell +\frac{n_{\mu}n^{\mu}}{2}\vec k \right|_{\varsigma}\hspace{1cm}
k_\mu n^\mu=1.
$$
Observe that the causal character of $\vec n$ is unrestricted.
The variation of the vanishing expansion $\theta^+ =0$ reduces to\cite{AMS,AMS1,BS1}
$$
\delta_{f\vec n} \theta^+=-\Delta_{\varsigma}f+f\left.\left(\frac{1}{r^2}-
G_{\mu\nu}k^\mu \ell^{\nu}-\frac{n_{\rho}n^\rho}{2}G_{\mu\nu}k^\mu k^\nu \right)
\right|_\varsigma
$$
where $\Delta_{\varsigma}$ is the Laplacian on $\varsigma$.
Now, choose $f=a_{0}+a_{N}P_{N}(\cos\theta)$, ($a_0,a_{N}=$ const.), where $P_{l}$ are the Legendre polynomials. Using $\Delta_\varsigma P_l =- \frac{ l (l +1)}{r_{\varsigma}^2}P_l$, the previous variation becomes
\bean
\delta_{f\vec n} \theta^+ =a_0 \left.\left(\frac{1}{r^2}-
G_{\mu\nu}k^\mu \ell^{\nu}-\frac{n_{\rho}n^\rho}{2}G_{\mu\nu}k^\mu k^\nu \right)
\right|_\varsigma \hspace{1cm}\\
+ {a_{N}\left.\left(\frac{1+N(N+1)}{r^2}-
G_{\mu\nu}k^\mu \ell^{\nu}-\frac{n_{\rho}n^\rho}{2}G_{\mu\nu}k^\mu k^\nu \right)
\right|_\varsigma\, P_N (\cos\theta)} .
\eean
The term in the last line can be made to vanish by choosing
$$
n_\rho n^\rho =\frac{2}{\left.G_{\mu\nu}k^\mu k^{\nu}\right|_\varsigma}
\left.\left(\frac{1+N(N+1)}{r^2} -G_{\mu\nu}k^\mu \ell^{\nu} \right)\right|_\varsigma \, .
$$
(Here is where one needs that $\left.G_{\mu\nu}k^\mu k^{\nu}\right|_\varsigma\neq 0$).
With this choice 
$$\delta_{f\vec n} \theta^+=-a_0 \frac{ N (N+1)}{r_{\varsigma}^2}$$ 
hence, the deformed surface is f-trapped for any $a_{0}>0$.
As $f=a_0 +a_N P_N (\cos\theta)$, setting $a_N<-a_0<0$ implies that $f$ is negative around $\theta =0$ and positive where $P_N \leq 0$. Thus, the deformed surface is f-trapped and enters both sides of AH.
\end{proof}
Therefore, $\B$ is not a spherically symmetric  MTT: $\B\neq \AH$. It follows that the spherically symmetric MTTs are in the f-trapped region,
$\AH\backslash \AH^{iso} \subset \mathscr{T}$, from where it also follows that the boundary itself cannot touch AH (except perhaps in isolated-horizon portions): $\B\cap (\AH\backslash \AH^{iso})=\emptyset $. This is independent of the causal character of AH. 

Based on the previous result, and with the aid of a very strong and fundamental result found in Ref.\refcite{AM} concerning the existence of MOTS between outer trapped and outer untrapped surfaces (see also Ref.\refcite{KH}), one can also derive a result on the existence of non-spherical MTTs, see Ref.\refcite{BS1}.
\begin{corollary}
In arbitrary spherically symmetric spacetimes there exist MTTs penetrating both sides of the apparent 3-horizon $AH: \{r=2m\}$ at any region where $G_{\mu\nu}k^\mu k^{\nu}|_{AH} \neq 0$.\label{coro}
\end{corollary}
Actually, all spacelike MTTs (that is, DHs) other than AH must lie partly to the future of AH and partly to its past, as proven in Ref.\refcite{AG}.
\begin{theorem}
No closed weakly f-trapped surface can be fully contained in the past domain of dependence $D^-(\AH)$ of a spacelike AH.
\label{th:AG}
\end{theorem}

However, closed f-trapped surfaces may lie on $D^-(\AH)$ almost completely. In other words, closed f-trapped surfaces can intersect the region $\{r<2m\}$ in just an arbitrarily tiny portion, as small as desired. This surprising result was obtained in Ref.\refcite{BS1} and will be of fundamental importance for the concept of ``core'' of a black hole, see section \ref{sec:cores}. More precisely:
\begin{theorem}
In spherically symmetric spacetimes, there are closed f-trapped surfaces (topological spheres) penetrating both sides of the apparent 3-horizon $\AH\backslash\AH^{iso}$ {with arbitrarily small portions} outside the region $\{ r>2m\}$.
\label{th:tiny}
\end{theorem}

\section{Closed trapped surfaces are ``clairvoyant''}
Now that we have learnt that closed f-trapped surfaces can extend beyond the spherically symmetric MTT, one can wonder haw far can they go. In particular, one can ask whether or not they can extend all the way `down' to the flat portions of the spacetime, if they exist. Again, the surprise is that they do extend and penetrate flat portions (under some circumstances). 

This was proven, via an explicit example, in Ref.\refcite{BS}, later refined with more elaborated examples in Ref.\refcite{ABS}. The chosen spacetime was a particular simple case of the Vaidya spacetime (\ref{vaidya}), defined by the explicit mass function
$$
m(v)=\left\{\begin{array}{cl}
0 & \hspace{3mm} v<0\\
\mu v &  \hspace{3mm} 0\leq v \leq M/\mu\\
M &  \hspace{3mm} v>\mu
\end{array}
 \right.
$$
where $\mu$ is a constant and $M$ is the total mass, also a constant.
Thus, this is flat for $v<0$, it ends in a Schwarzschild region with mass $M$ ($v>M/\mu$), and it happens to be self-similar in the intermediate Vaidya region for $0<v<M/\mu$.

{The closed f-trapped surface is composed of:}

\begin{itemize}
\item {Flat region}: a topological disk given by the hyperboloids
$$
\theta =\pi/2 ; \hspace{2cm} v=t_0 +r -\sqrt{r^2+k^2} 
$$

\item {Vaidya region}: a topological cylinder defined by $\theta =\pi/2$ and
$$
\sqrt{v^2-bvr+ar^2}=C\exp \left\{\frac{b}{2\sqrt{a-b^2/4}}\arctan\left(\frac{2v-br}{2r\sqrt{a-b^2/4}} \right)\right\}
$$
with $a>b^2/4$ and $C=$constant. ($dv/dr\rightarrow \infty$ at $v=br$).

\item {Schwarzschild region}: another disk composed of two parts
\begin{itemize}
\item a cylinder with $
\theta =\pi/2$ and $r=\gamma M$(=const.)

\item a final ``capping" disk defined by ($r=\gamma M$)
$$
\left(\theta -\frac{\pi}{2} +\delta \right)^2 +\left(\frac{v}{\gamma M}-c_1 \right)^2=\delta^2
$$
with  constants $c_1$ and $\delta$.
\end{itemize}
\end{itemize}
\begin{figure}[h!]
\begin{center}
\includegraphics[width=7.2cm]{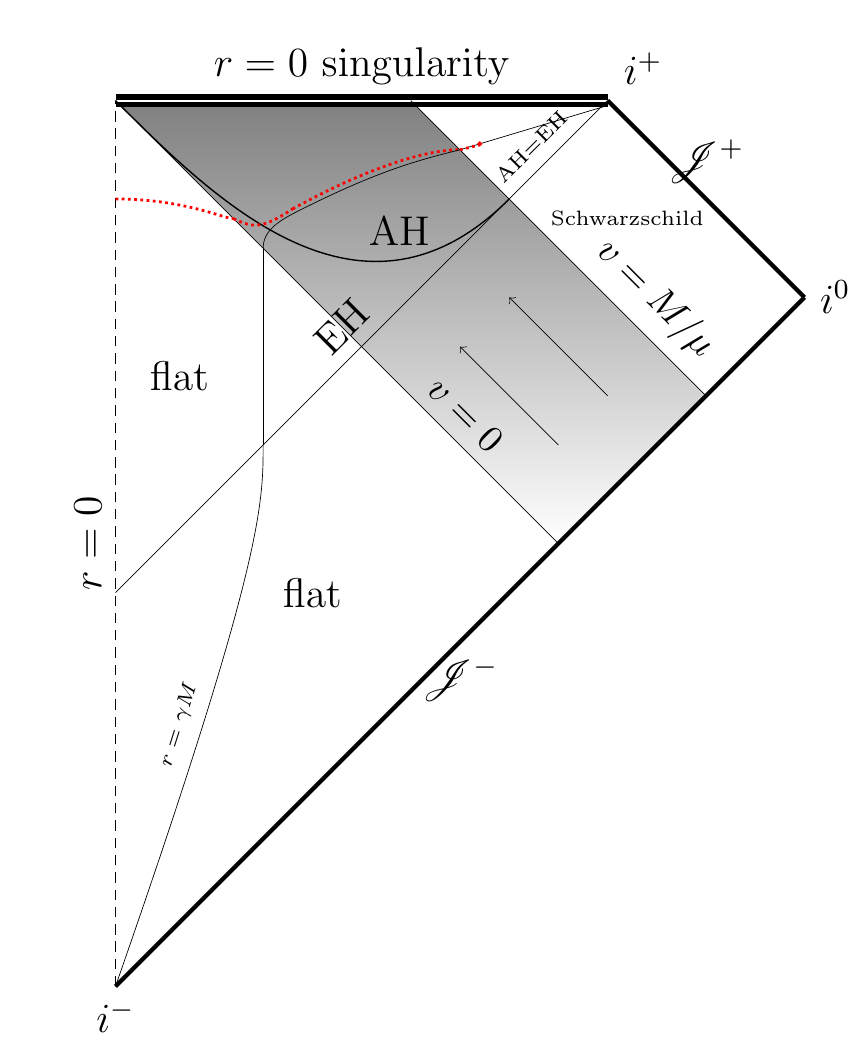}
\end{center}
\caption{A closed f-trapped surface penetrating the flat portion of a black hole spacetime constructed with the Vaidya solution. The surface has points on every sphere on the dotted line in red. It thus enters the flat portion of the spacetime.}
\label{f8}
\end{figure}
This surface is future-trapped if (1) $t_0<k,\, k>0$, (2), $0<a<b$, (3) $1>\gamma=\frac{1}{b\mu}$,  $a\geq \frac{1}{\mu}$, (4) $0<\delta\leq \frac{\pi}{2}$ and $ \sqrt{\frac{2}{\gamma-1}}\left(\frac{1}{\gamma}-1\right) >\frac{1}{\delta}$.
Note that these conditions imply in turn a restriction on the mass growth, as the slope of the mass function must satisfy
$$
\mu=\frac{1}{\gamma b}>\frac{1}{\gamma}\frac{b}{4a}>\frac{1}{4\gamma} \hspace{1cm} \gamma < 0.68514 \, .
$$
A picture showing points of the surface in a Penrose diagram is shown in figure \ref{f8}.

In conclusion, closed trapped surfaces may panetrate flat regions of the spacetime. This implies that they are highly non-local: they can have portions in a flat region of spacetime whose {\em whole} past is also flat in clairvoyance of energy that crosses them elsewhere to make their compactness feasible. They are {\em clairvoyant}, that is to say, `aware' of things that happen elsewhere, far away, with no causal connection.

Observe finally that this result also implies that the boundary $\B$ of the f-trapped region may penetrate the flat regions too.

\section{Interplay of surfaces and generalized symmetries}
Hitherto, we have proven that the boundary $\B$ of the f-trapped region must be generically below the spherically symmetric MTT defined by AH: $\{r=2m\}$, but it cannot extend all the way down to the event horizon EH.
In order to try to put restrictions on the location of $\B$ and the extension of $\mathscr{T}$ we use the following fundamental property\cite{BS1}
\begin{proposition}\label{prop:1}
No f-trapped surface (closed or not) can touch a spacelike hypersurface to its past at a single point, or have a 2-dimensional portion contained in the hypersurface, if the latter has a positive semi-definite second fundamental form.
\end{proposition}
Before giving the main steps for the proof of this basic result, the intuitive idea behind it can be understood by means of the following figure \ref{f9}. As we saw in Fig.\ref{f3}, f-trapped surfaces (closed or not) have to bend down in ``time'', if this time has some appropriate properties. In particular, if the time defines a foliation by hypersurfaces with positive semi-definite second fundamental forms, then Fig.\ref{f9} shows that there is no way that the surface can touch tangentially, to its past (in ``time''), any such level hypersurface.
\begin{figure}[h!]
\begin{center}
\includegraphics[height=3.3cm]{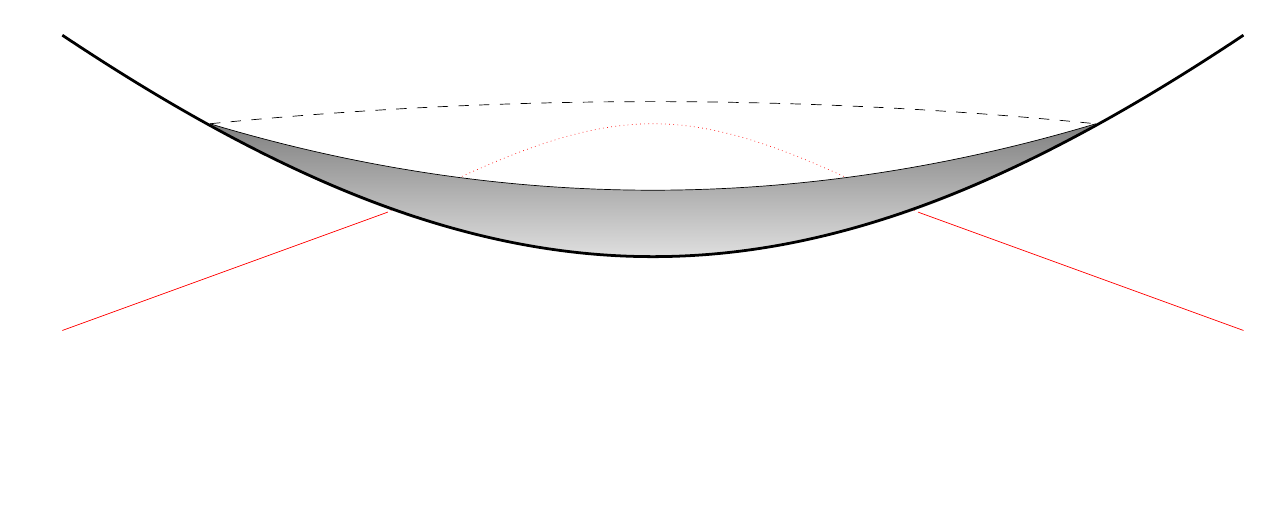}
\end{center}
\caption{A (hyper)surface with positive-definite second fundamental form (shown as a hyperboloid) and a f-trapped curve (shown in red) in 3-dimensional Minkowski spacetime.}
\label{f9}
\end{figure}
Thus, one can use (generalized) symmetries, or equivalently some distinguished hypersurface-orthogonal future-pointing vector fields ---such as (conformal, homothetic) Killing vectors, Kerr-Schild vector fields\cite{CHS},  the Kodama vector field in spherically symmetric spacetimes, etcetera---
to constrain the possible existence of f-trapped surfaces. 
This interplay between surfaces and generalized symmetries has proven very useful in several investigations concerning the trapped-surface fauna\cite{BS1,CM,MS,S2,S3,S5}.

To prove the fundamental property given in Proposition \ref{prop:1} and related interesting results, start with the identity for the Lie derivative of the metric along a vector field $\vec \xi$
$$(\lie_{\xiv} g)_{\mu\nu}=\nabla_{\mu}\xi_{\nu}+\nabla_{\nu}\xi_{\mu}.$$
Projecting to the surface $S$ and using (\ref{proj})
$$
(\lie_{\xiv} g|_{S})_{\mu\nu}\, e^{\mu}_Ae^{\nu}_B=\overline\nabla_{A}\overline\xi_{B}+
\overline\nabla_{B}\overline\xi_{A}+2 \xi_{\mu}|_{S} K^{\mu}_{AB}
$$
so that contracting now with $\gamma^{AB}$ we arrive at the main formula
\be
\fbox{$\displaystyle{\frac{1}{2}
P^{\mu\nu}(\lie_{\xiv} g|_{S})_{\mu\nu} =
\overline\nabla_{C}\overline\xi^{C}+ \xi_\rho H^\rho}$}
\label{main}
\ee
where  $P^{\mu\nu}\equiv \gamma^{AB} e^{\mu}_Ae^{\nu}_B$
is the orthogonal projector of $S$. 

This elementary formula is very useful and permits to obtain many interesting, immediate, consequences. For instance: 
\begin{enumerate}
\item If $S$ is minimal ($\vec H =\vec 0$), integrating the formula for closed $S$
$$
\oint_{S}P^{\mu\nu}(\lie_{\xiv} g|_{S})_{\mu\nu}=0 \, .
$$
This relation must be satisfied for {\em all} possible vector fields $\xiv$. Therefore, closed minimal surfaces are very rare.
\item If $\xiv$ is a Killing vector, integrating again for closed $S$
$$
\oint_{S}\xi_\rho H^\rho =0 \, .
$$
Therefore, if the Killing vector $\xiv$ is timelike on $S$, then $S$ cannot be weakly
trapped, unless it is minimal\cite{MS}.
\end{enumerate}
More sophisticated and useful consequences can be derived, such as\cite{BS1,MS}
\begin{lemma}
Let $\xiv$  be future-pointing on a region $\R\subset \varietat$ and let $S$ be a closed surface contained in $\R$ with $P^{\mu\nu}(\lie_{\xiv} g|_{S})_{\mu\nu} \geq 0$. Then, $S$ cannot be closed and weakly f-trapped ({unless} $\xi_{\mu}H^\mu =0$ and $P^{\mu\nu}(\lie_{\xiv} g|_{S})_{\mu\nu} = 0$.)
\label{lemma}
\end{lemma}
\begin{proof}
Integrating the main formula on $S$, the divergence term integrates to zero so that
$$
\oint_{S}\xi_\rho H^\rho =\frac{1}{2} \oint_S P^{\mu\nu}(\lie_{\xiv} g|_{S})_{\mu\nu} \geq 0
$$
Hence, $\vec H$ cannot be future pointing on all $S$ (unless $\xi_\mu H^\mu=P^{\mu\nu}(\lie_{\xiv} g|_{S})_{\mu\nu} =0$.)
\end{proof}

Stronger results can be obtained for {\em hypersurface-orthogona}l $\xiv$, that is, $\xiv$ satisfying
$$
\xi_{[\mu}\nabla_{\nu}\xi_{\rho]}=0 \hspace{2mm} \Longleftrightarrow \hspace{2mm} \xi_{\mu}=-F \partial_{\mu} \tau
$$
for some local functions $F>0$ and $\tau$. This means that
$\xiv$ is orthogonal to the hypersurfaces $\tau =$const.
(called the {level hypersurfaces}.)

\begin{theorem}
Let $\xiv$ be future-pointing and hypersurface-orthogonal on a region $\R\subset \varietat$ and let
$S$ be a f-trapped surface. Then, $S$ cannot have a local minimum of $\tau$ at any point $q\in \R$ where $P^{\mu\nu}(\lie_{\xiv} g)_{\mu\nu}|_q\geq 0$.
\label{th:no-min}
\end{theorem}
\begin{proof}
Let $q\in S\cap \R$ be a point where $S$ has a local extreme of $\tau$. 
Noting that $\bar\xi_A=-\bar F\partial\bar\tau/\partial \lambda^A$ with $\bar F \equiv F|_S$ this means that
$$
\bar\xi_A|_q= \left.\frac{\partial\bar\tau}{\partial \lambda^A}\right|_q=0
$$
where $\tau =\bar\tau(\lambda^A)$ is the local parametric expression of $\tau$ on $S$. 

An elementary calculation leads then to:
$$
\left.\overline\nabla_{A}\overline\xi^{A}\right|_q=\left.\gamma^{AB}\overline\nabla_A\left(-\bar F\frac{\partial\bar\tau}{\partial\lambda^B}\right)\right|_q=\left.-\bar F\gamma^{AB}\frac{\partial^2\bar\tau}{\partial\lambda^A\partial\lambda^B}\right|_q
$$
Introducing this in the main formula (\ref{main}) 
$$
\left.\bar F\gamma^{AB}\frac{\partial^2\bar\tau}{\partial\lambda^A\partial\lambda^B}\right|_q=\left.-\frac{1}{2}
P^{\mu\nu}(\lie_{\xiv} g)_{\mu\nu} \right|_q+
\left.\xi_\rho H^\rho\right|_q\leq \left.\xi_\rho H^\rho\right|_q
$$
hence $\partial^2\bar\tau/\partial\lambda^A\partial\lambda^B|_q$ cannot be positive (semi)-definite if $\xiv$ and $\vec H$ are both future-pointing. 
\end{proof}
\begin{remark}
\begin{enumerate}
\item $S$ does not need to be compact, nor contained in $\R$.
\item It is enough to assume $P^{\mu\nu}(\lie_{\xiv} g)_{\mu\nu}|_q\geq 0$ only at the local extremes of $\tau$ on $S$. 
\item A positive semi-definite $\partial^2\bar\tau/\partial\lambda^A\partial\lambda^B|_q$ is also excluded.
\item The theorem holds true for weakly trapped surfaces with the only exception of
$$
\left.\frac{\partial^2\bar\tau}{\partial\lambda^A\partial\lambda^B}\right|_q=0 \hspace{3mm} \mbox{{and}} \hspace{3mm}
P^{\mu\nu}(\lie_{\xiv} g)_{\mu\nu}|_q=0 \hspace{3mm} \mbox{{and}}  \hspace{3mm} \left.\xi_\rho H^\rho\right|_q=0 
$$
If $\xiv|_q$ is timelike, the last of these implies that $\vec H|_q=\vec 0$. 
\item Letting aside this exceptional possibility, $\tau$ always decreases at least along one tangent direction in $T_qS$. 
Starting from any point $x\in S\cap\R$ one can always follow a connected path along $S\cap\R$ with decreasing $\tau$.
\end{enumerate}
\end{remark}
\section{The past barrier $\S$}
The results above on the interplay of vector fields with special convexity properties and the trapped-surface fauna were essential in order to detect a general past barrier for closed f-trapped surfaces in general imploding, asymptotically flat, spherically symmetric spacetimes\cite{BS1}. The mild
assumptions used to get this past barrier are\cite{BS1}:
\begin{enumerate}
\item The total mass function is finite, and there is an initial  flat region
$$
m(v,r)=0 \hspace{3mm} \forall v<0; \hspace{1cm} \forall v>0:\hspace{3mm}
0\leq m(v,r)\leq M <\infty \,\,\, (M>0)
$$
and a regular future null infinity $\scri^+$ with associated event horizon EH.\cite{D} Let 
$\AH_1$ denote the connected component of $\AH \equiv \{r=2m\}$ associated to this EH. 
$\AH_1$ separates the region $\R_1$, defined as the connected subset of $\{r>2m\}$ which contains the flat region of the spacetime, from a region containing f-trapped 2-spheres. 
\item The dominant energy condition holds, and furthermore the matter-energy is falling in
$$
\frac{\partial m}{\partial v}\geq 0 \hspace{5mm} \mbox{on}\hspace{5mm}  \{r\geq 2m\}\cap J^+(EH) \label{dotm} 
$$
where $J^+(EH)$ is the causal future of the event horizon.
\end{enumerate}

The connected component of $\B$ associated to $\AH_1$ will be denoted by $\B_1$, in analogy with $\R_1$ which denotes the corresponding $\{r>2m\}$-region. 
The same notation is used for $\AH_1^{iso}$, $\mathscr{T}_1$, etcetera.
Some examples are provided in the Penrose diagrams of Fig.\ref{f10}.
\begin{figure}[h!]
\includegraphics[width=15cm]{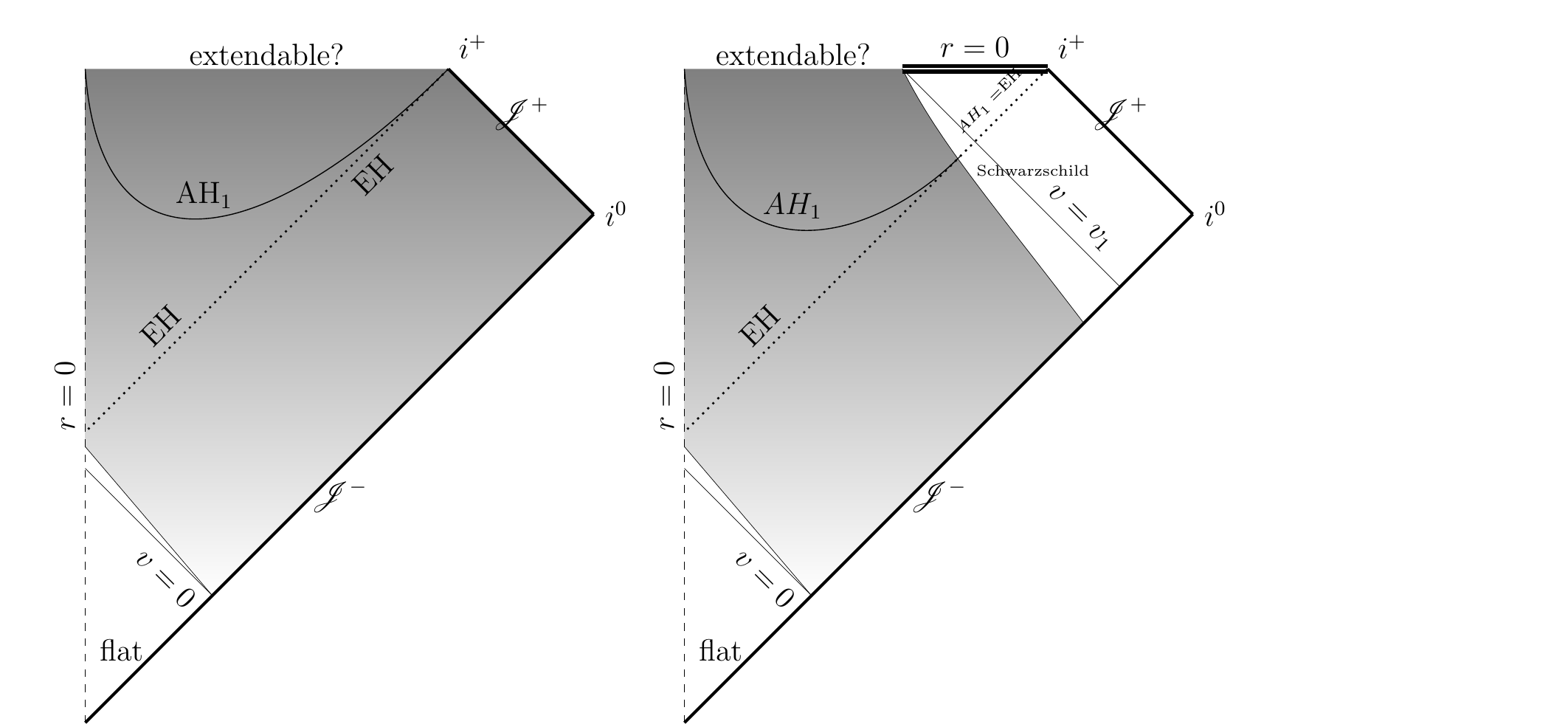}
\caption{Two conformal diagrams that correspond to the possibilities $m(v,r)<M$ everywhere (left), and to $m(v,r)=M$ is some asymptotic region (right). In the second case the spacetime becomes eventually Schwarzschild spacetime. The shaded regions are non-flat spherically symmetric spacetimes with non-zero energy-momentum tensor. A connected component of AH, labeled $\AH_{1}$, is shown. This may have timelike portions, and it merges with the EH either asymptotically (left) or at a finite value of $v$ (right). Whether or not a singularity develops in the shaded region depends on the specific properties of the matter (of the mass function $m(v,r)$), and there can be cases where the spacetime can be continued towards the future, where some other connected parts of AH may appear.}
\label{f10}
\end{figure}

\subsection{The Kodama vector field} 
The Kodama vector field\cite{Ko}, which in these coordinates takes the simple form
$$
\xiv =e^{-\beta}\partial_v
$$%
characterizes the spherically symmetric directions tangent to the hypersurfaces $r=$const. In other words, it points into the unique direction where the round spheres have vanishing expansion, as can be checked using (\ref{Hspheres}): $\xi_\mu H^\mu_{sph}=0$.

$\xiv$ is hypersurface orthogonal, with the level function $\tau$ defined by
$$
\xi_{\mu}dx^\mu =-Fd\tau = dr-e^{\beta} \left(1-\frac{2m(v,r)}{r}\right) dv
$$
Furthermore
$$
\xi_\mu\xi^\mu=-\left(1-\frac{2m(v,r)}{r}\right), \hspace{1cm} \ell_\mu\xi^\mu=-1
$$
so that $\xiv$ is future-pointing timelike on the region $\{r<2m\}$, and future-pointing null at $\AH :\{r=2m\}$. Thus, in order to ascertain if $\xiv$ has all the necessary properties to apply Theorem \ref{th:no-min} we need to check if
$$\left.P^{\mu\nu}(\lie_{\xiv} g|_{S})_{\mu\nu} \right|_q\geq 0$$
at any point $q\in S\cap \{r\geq 2m\} \cap J^+(EH)$ where $S$ has a local extreme of $\tau$. 
The Lie derivative can be easily computed 
$$
(\lie_{\xiv} g)_{\mu\nu} =e^{\beta}\frac{2}{r}\frac{\partial m}{\partial v}\ell_\mu \ell_\nu -\frac{\partial\beta}{\partial r}\left(\delta_{\mu}^r \xi_{\nu}+\delta_\nu^r \xi_\mu \right)
$$
Then, given that $\bar\xi_A|_q=0$ we obtain
$$
\left.P^{\mu\nu}(\lie_{\xiv} g|_{S})_{\mu\nu} \right|_q=\left.e^{\beta}\frac{2}{r}\frac{\partial m}{\partial v}\right|_S \bar\ell_A \bar\ell^A \geq 0
$$
as required.

\subsection{A past barrier for closed f-trapped surfaces}
The hypersurfaces $\tau=\tau_c$=const.\  are spacelike everywhere (and approaching $i^0$) if $\tau_c<\tau_{\Sigma}$, while they are partly spacelike and partly timelike, becoming null at $\AH_1$, if $\tau_c>\tau_{\Sigma}$, where
$$
\tau_\Sigma\equiv \inf_{x\in AH_1} \tau |_x
$$
Observe that $\tau_\Sigma$ is also the least upper bound of $\tau$ on EH.

Define the hypersurface $\Sigma$ as
$$\Sigma\equiv \{\tau =\tau_\Sigma\}\, .$$
$\Sigma$ is the {\em last} hypersurface orthogonal to $\xiv$ which is non-timelike everywhere. See figure \ref{f11} for a representation of all these facts.
It turns out that $\Sigma$ is a past limit for closed f-trapped surfaces, and this is a direct consequence of the properties of the Kodama vector field and Theorem \ref{th:no-min} (or Proposition \ref{prop:1}. Thus\cite{BS1}
\begin{theorem}
No closed f-trapped surface can enter the region $\tau\leq \tau_\Sigma$.
\end{theorem}
The location of $\Sigma$ acquires therefore an unexpected importance, and this depends in particular on whether $8\dot m_0>(1-2m'_0)^2$ or not. 
Here $\dot m_0$ and $m'_0$ are the limits of $\frac{\partial m}{\partial v}$ and $\frac{\partial m}{\partial r}$ when approaching $r=0$, respectively.
In the former case, $\Sigma$ does penetrate the flat region. It may not be so in the other cases. See Ref.\refcite{BS1} for details. 
\begin{figure}[!ht]
\includegraphics[width=15cm]{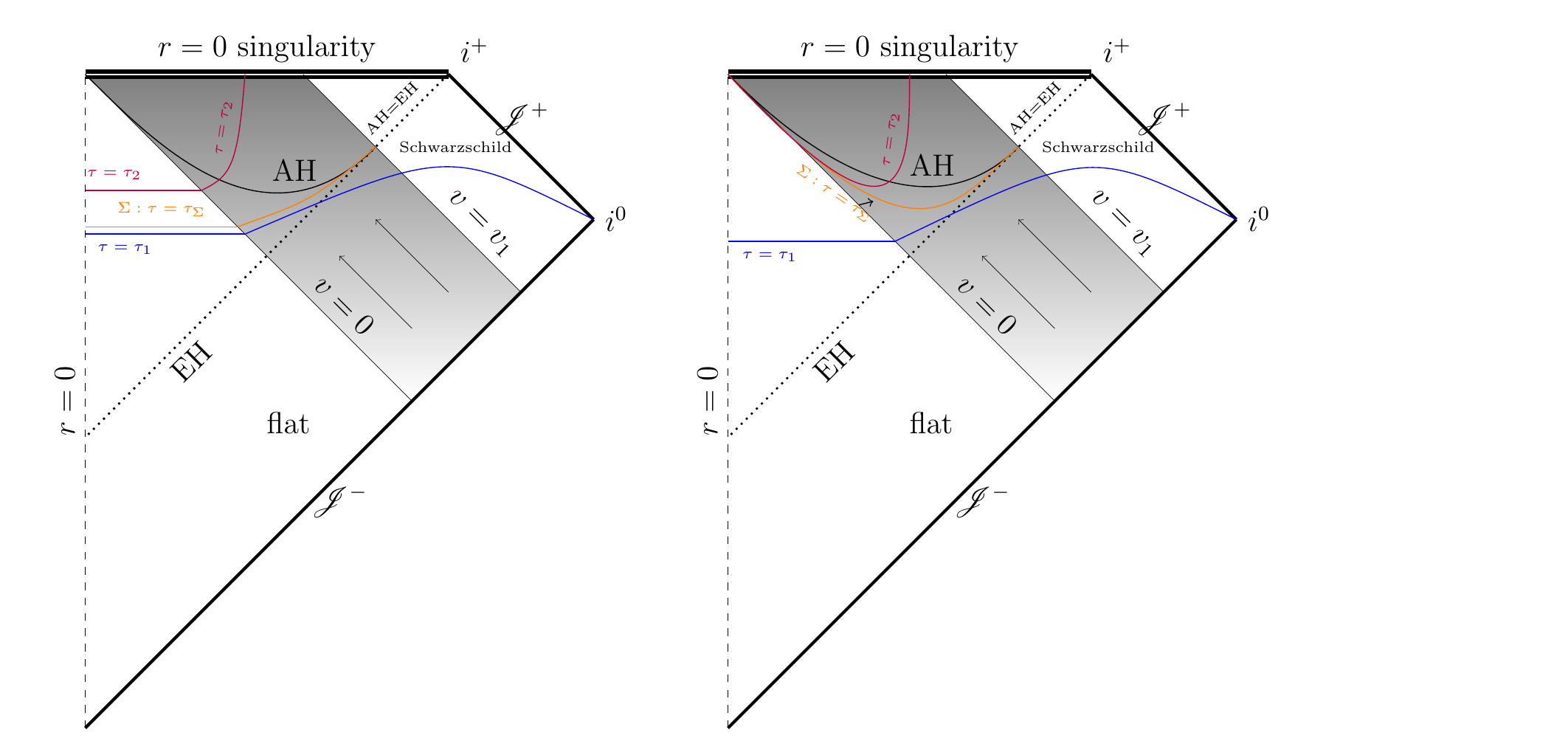}
\caption{The past barrier $\Sigma$ represented for some particular cases of the Vaidya spacetime (\ref{vaidya}). The hypersurfaces $\tau=\tau_1$ are spacelike for $\tau_1<\tau_\S$, while they are partly spacelike, partly timelike for values of $\tau_2> \tau_\S$, as shown in the figure. The limit case defines $\S:\, \tau=\tau_\S$, which happens to be a past barrier for f-trapped surfaces due to the convexity properties of the Kodama vector field. This past barrier $\S$ can enter the flat region or not, and this depends on the properties of the mass function close to the upper left corner with $r=0$. Here, two possibilities are depicted. The case where $\S$ never enters the flat region of the spacetime (right) and the opposite case where it does penetrate the flat region (left).}
\label{f11}
\end{figure}

\section{Some properties of $\B$, and about its location}
Some elaborations using the previous results on the Kodama level function $\tau$ allows one to derive further properties of the boundary $\B$, and to put severe restrictions on its location. To that end,  set $\tau_\B\equiv \inf_{x\in \B} \tau |_x$
where $\tau =$const. are the level hypersurfaces of $\xiv$.  Then the following set of results were obtained in Ref.\refcite{BS1}, where the reader may consult the proofs.
\begin{proposition}
The connected component $\B_1$ cannot have a positive minimum value of $r$, and furthermore
$$
\tau_\B = \inf_{x\in \B_1} \tau |_x=\tau_\Sigma \, .
$$
\end{proposition}
\begin{corollary} $\B_1\subset (\R_1\cup \AH_1) \cap \{\tau\geq \tau_\S\}$
and $\B_1$ must merge with, or approach asymptotically, $\Sigma$, $\AH_1$ and EH in such a way that $\B_1\cap (\AH_1\backslash \AH_1^{iso}) =\emptyset$
if $G_{\mu\nu}k^\mu k^\nu |_{AH_1\backslash EH}>0$. 
 
Furthermore, $\B_1$ cannot be non-spacelike everywhere.
\end{corollary}

\begin{theorem}
$(\B\backslash \AH_1^{iso})$ cannot be a marginally trapped tube, let alone a dynamical horizon.
Actually, $(\B\backslash \AH_1^{iso})$ does not contain any non-minimal closed weakly f-trapped surface.
\end{theorem}
This theorem came as a surprise because there was a spread belief, due specially to some convincing arguments by Hayward\cite{Hay}, on the contrary. As clearly explained in Ref.\refcite{AK1}, the convincing arguments were based on assumptions that seemed quite natural intuitively, but that were very strong technically. It turns out that these assumptions, despite looking intuitively natural, were almost never met so that the intended derived result was essentially empty. 

Another consequence of Theorem \ref{th:no-min} is that 
$\B_1\backslash \AH_1^{iso}$ has to bend down in Kodama ``time'' $\tau$ (in the region $\{r>2m\}$ where the level function $\tau$ is a timelike coordinate), in analogous manner to what was shown in Fig.\ref{f3} for Minkowskian time.
\begin{proposition}
$\tau$ is a nonincreasing function of $r$ on any portion of the connected
component $\B_1$ which is locally to the past of $\mathscr{T}_1$, and it is actually strictly decreasing at least somewhere on $\B_1\backslash \AH_1^{iso}$.
In particular, 
$\B_1\cap (\Sigma\backslash EH)=\emptyset$.
\label{prop}
\end{proposition}
\begin{figure}[h!]
\begin{center}
\includegraphics[height=6cm]{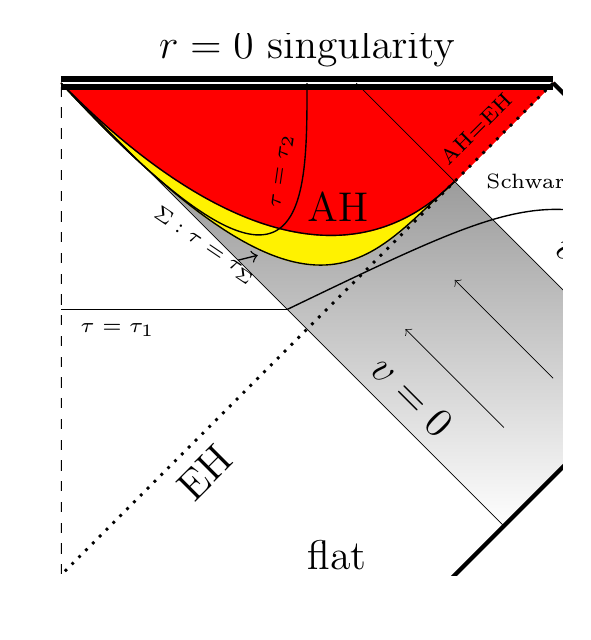}
\end{center}
\caption{This is an enlargement of the conformal diagram on the right in Fig.\ref{f11}, showing the allowed region for the location of the boundary $\B$. The part shown in red is known to be part of the f-trapped region $\mathscr{T}$. Theorem \ref{th:zigzag} tells us, however, that $\mathscr{T}$ extends beyond AH, and actually includes it. On the other hand, Proposition \ref{prop} implies that $\mathscr{T}$ can never touch $\S$ (outside EH). Therefore, the zone marked in yellow is the allowed region for the boundary $\B$, keeping in mind that it must be strictly inside the yellow zone: it can never touch $\S$, nor $\AH$, there.}
\label{f12}
\end{figure}
This result infers that $\B_1$ is to the future of $\Sigma$, and we already know that is has to be to the past of $\AH_1$. This also implies that
$\B_1$ is squeezed below by $\S$ and above by $\AH_1$ close to their merging with the event horizon, and thus $\B_1$ must be spacelike close to its merging with $\Sigma$, $\AH_1$ and EH.\cite{BS1}
An illuminating pictorial explanation of these results is represented, for a particular case, in Fig.\ref{f12}.

There remains, as an interesting puzzle, the question of where is {\em exactly} the boundary $\B$. This is an open question. There are some known restrictions on the 2nd fundamental form (extrinsic curvature) of $\B$\cite{BS1}, but they are not sufficient to completely determine the position of $\B$ in generic spherically symmetric spacetimes.

The more restrictive known property on $\B$ is given by the following result\cite{BS1}
\begin{proposition}
Any spacelike portion of the connected component $\B_1$ which is
locally to the past of $\mathscr{T}_1$ has a second fundamental
form with a non-positive (and strictly negative whenever $\B_1$ is not tangent to a $\tau$=const hypersurface) double
eigenvalue. In particular,
it cannot have a positive semidefinite second fundamental
form there.
\end{proposition}
Actually, one can find stronger restrictions on the extrinsic properties of $\B$, but they are out of the scope of this contribution.

\section{Black Holes' Cores}
\label{sec:cores}
Let me summarize some of the main conclusions derived so far concerning trapped surfaces, MTTs, EH, $\B$ and the trapped region $\mathscr{T}$. 

The clairvoyance property of trapped surfaces is inherited by everything based on them, 
such as marginally trapped tubes including dynamical horizons. It implies that EH is teleological and that $\B$ is also non-local, even penetrating sometimes into flat portions of the spacetime.
In conjunction with the non-uniqueness of dynamical horizons, this poses a fundamental puzzle for the physics of black holes.
Four possible solutions have been put forward\cite{AK1,BS1,B,GJ,NJKS,Kri} 
\begin{enumerate}
 \item one can rely on the old and well defined event horizon, and try to put up with its teleological properties. 
 \item one can accept the non-uniqueness of MTTs and treat all possible MTTs and dynamical/trapping horizons on equal footing. 
 \item one can also use the boundary $\B$ as defined above, despite its non-local properties. 
 \item or one can try to define a preferred  MTT. Hitherto, the only proposal I am aware of was presented in Ref.\refcite{GJ}, based on an evolution principle for the area (entropy) of the marginally trapped surfaces foliating the MTT.
\end{enumerate}

In Ref.\refcite{BS1} we have put forward a novel strategy. The idea is based on the simple question: 
{\em what part of the spacetime is absolutely indispensable for the existence of the black hole?}
Surely enough, any flat region is certainly not essential for the existence of the black hole.
What is?
\begin{definition}
A region $\mathscr{Z}$ is said to be a {\em core} of the f-trapped region $\mathscr{T}$ 
if it is a minimal closed connected set that needs to be removed from the spacetime in order to get rid of all closed f-trapped 
surfaces in $\mathscr{T}$, 
and such that any point on the boundary $\partial\mathscr{Z}$ is connected to $\B=\partial \mathscr{T}$ 
in the closure of the remainder.
\end{definition}
\begin{remark}
Here, ``minimal" means that there is no other set $\mathscr{Z}'$ with the same properties and properly contained 
in $\mathscr{Z}$. 
The final technical condition is needed because one could identify a particular removable region 
to eliminate the f-trapped surfaces, 
excise it, but then put back a tiny but central isolated portion to make it smaller. 
However, this is not what one wants to cover with the definition.  
\end{remark}
Obviously, $\mathscr{Z}\subset \mathscr{T}$, but in general $\mathscr{Z}$ is substantially smaller than the corresponding trapped region $\mathscr{T}$. An example of a core is given by the dust Robertson-Walker model of Fig.\ref{f3}, where the region shown in purple is a core of the larger red region $\mathscr{T}$. If the purple region is removed from the spacetime, then no closed f-trapped surface remains. This example demonstrates that cores are not unique: 
one can choose any other region $\mathscr{Z}$ equivalent to the chosen one by moving all its points by the group of symmetries on each homogeneous spatial slice of the Robertson-Walker metric.
This kind of non-uniqueness is somehow irrelevant, being due to the existence of a high degree of symmetry. Nevertheless, even in less symmetric cases the uniqueness of the cores cannot be assumed beforehand, and one can actually prove that it does not hold in general, see Proposition \ref{prop:} below.

\begin{theorem}
In spherically symmetric spacetimes (with $\AH^{iso}\backslash EH=\emptyset$)  the region 
$$\mathscr{Z}\equiv \{r\leq 2m(v,r)\}$$
 is a core of the f-trapped region. 
\end{theorem}
This follows firstly from the fact that removing $\mathscr{Z}$ from the spacetime short-circuits all possible closed f-trapped surfaces, as they cannot be fully contained in the region where the Kodama vector field is timelike (the complement of $\mathscr{Z}$) as a consequence of Lemma \ref{lemma}. And secondly and more importantly, from Theorem \ref{th:tiny}, which implies that one cannot remove a smaller region achieving the same goal. (This is where we need the condition of not having isolated-horizon portions in AH, but this is probably technical and the result will hold in general).

The identified cores happen to be unique with spherical symmetry.
\begin{proposition}
In spherically symmetric spacetimes (with $\AH^{iso}\backslash EH=\emptyset$) 
$\mathscr{Z}=\{r\leq 2m\}$ are the only spherically symmetric cores of $\mathscr{T}$. 
Therefore, $\partial\mathscr{Z}=\AH$ are the only spherically symmetric boundaries of a core.
\end{proposition}
The two previous results are surprising and may have a deep meaning, because the trapped regions  and their cores are global concepts, and in that sense they share the teleological and/or clairvoyant properties of EH and of trapped surfaces. However, we have identified at least one boundary of a core which happens to be a MTT ---and a very good one in spherical symmetry: the unique one respecting the symmetry---, and MTTs are quasi-local objects, they do not need to know future causes or to be aware of things that happen elsewhere. A full interpretation of this surprising result may lead to a better understanding of BHs and of their boundaries.

It arises as an important problem the question of the uniqueness of cores, and their boundaries. As mentioned before, cores are not unique and one can prove the existence of 
non-spherically symmetric cores in spherically symmetric spacetimes\cite{BS1}.
\begin{proposition}\label{prop:}
There exist non-spherically symmetric cores of the f-trapped region in spherically symmetric spacetimes. 
\end{proposition}
The proof of this result is essentially based on a theorem\cite{AG} analogous to Theorem \ref{th:AG} but valid for general DHs, because then one derives that there must be a subset of the future of any dynamical horizon which, when removed from the spacetime, gets rid of all closed f-trapped surfaces.
However, as we do not have an analogous of Theorem \ref{th:tiny} for general DHs, it is unkown whether the core is a proper subset of, or the whole, future of the DH. 
Therefore, only two things may happen. 
\begin{enumerate}
\item For a generic MTT $H$, its causal future $J^+(H)$ is 
a core. This will amount to saying that MTTs are generically boundaries of cores for BHs.
\item Any MTT $H$ other than AH is such that its causal future $J^+(H)$ is 
{\em not} a core ---the core being a proper subset of $J^+(H)$. Hence, the identified core $\mathscr{Z}=\{r\leq 2m\}$ is special in the sense that its boundary $\partial\mathscr{Z}=\AH$ 
is a marginally trapped tube. Thereby, $\AH: \{r=2m(v,r)\}$ would be selected as the unique MTT which is the boundary of a core of the 
f-trapped region $\mathscr{T}$. 
\end{enumerate}

\section*{Acknowledgements}
I thank the organizers of the School for their kind invitation, and especially Prof. Roh-Suan Tung for making my visit possible. Some parts of this contribution are based on a fruitful collaboration with I. Bengtsson. Supported by grants
FIS2010-15492 (MICINN), GIU06/37 (UPV/EHU) and P09-FQM-4496 (J. Andaluc\'{\i}a---FEDER).


\end{document}